\documentclass[11pt]{amsart}
\usepackage{amsmath,bm}
\usepackage{amsthm,amssymb,dsfont}
\usepackage{url}		
\usepackage{xspace}		
\usepackage{booktabs}

\usepackage{rotating}	
\usepackage{pdflscape}
\usepackage{fancyvrb}

\usepackage{bbm}			

\usepackage[pdftex,                %
bookmarks         = true,
bookmarksnumbered = true,
pdfpagemode       = None,
pdfstartview      = FitH,
pdfpagelayout     = SinglePage,
colorlinks        = true,
linkcolor		= magenta, 
citecolor		= blue,
urlcolor          = red,
pdfborder         = {0 0 0}
]{hyperref}

\usepackage{graphicx}

\makeatletter
\newcommand{\addresseshere}{%
  \enddoc@text\let\enddoc@text\relax
}
\makeatother

\usepackage{amsfonts}
\usepackage{latexsym}
\usepackage{amssymb}
\usepackage{graphicx}
\usepackage{a4wide}

\newcommand{\wh}{\widehat}

\newcommand{\R}{\mathbb{R}}
\renewcommand{\l}{\ell}
\renewcommand{\P}{\mathbf{P}}
\newcommand{\cH}{\mathcal{H}}

\newcommand{\PBin}{\mathbf{PBin}}
\newcommand{\Bin}{\mathbf{Bin}}

\newcommand{\ol}{\overline} 
\newcommand{\wt}{\widetilde}

\newcommand{\BH}{[\textnormal{BH}]} 

\newcommand{\LR}{[\textnormal{LR}]} 
\newcommand{\GR}{[\textnormal{GR}]} 
\newcommand{\PB}{[\textnormal{PB}]}
\newcommand{\HLR}{[\textnormal{HLR}]} 
\newcommand{\HGR}{[\textnormal{HGR}]} 

\newcommand{\WLR}{[\textnormal{wLR}]} 
\newcommand{\WLRAM}{[\textnormal{wLR-AM}]} 
\newcommand{\WLRGM}{[\textnormal{wLR-GM}]} 
\newcommand{\WPBAM}{[\textnormal{wPB-AM}]}
\newcommand{\WPBGM}{[\textnormal{wPB-GM}]} 
\newcommand{\WGR}{[\textnormal{wGR}]} 
\newcommand{\WGRAM}{[\textnormal{wGR-AM}]}
\newcommand{\WGRGM}{[\textnormal{wGR-GM}]}

\newcommand{\DLR}{[\textnormal{DLR}]} 
\newcommand{\DGR}{[\textnormal{DGR}]} 
\newcommand{\DPB}{[\textnormal{DPB}]}

\newcommand{\FDR}{\textnormal{FDR}}
\newcommand{\FDP}{\textnormal{FDP}}
\newcommand{\FDX}{\textnormal{FDX}}
\newcommand{\FWER}{\textnormal{FWER}}

\newcommand{\FAM}{F^{\textnormal{AM}}}
\newcommand{\FGM}{F^{\textnormal{GM}}}

\newtheorem{theorem}{Theorem}[section]

\newtheorem{proposition}{Proposition}[section]

\newtheorem{definition}{Definition}[section]
\newtheorem{remark}{Remark}[section]
\newtheorem{lemma}{Lemma}[section]
\usepackage{natbib}

\begin{document}
\title[Controlling false discovery exceedance for heterogeneous tests]{Controlling false discovery exceedance for heterogeneous tests}
\author{Sebastian D\"ohler and Etienne Roquain}
\date{\today}
\begin{abstract}
Several classical methods exist for controlling the false discovery exceedance (FDX) for large scale multiple testing problems, among them the Lehmann-Romano procedure \cite{LR2005} ($\LR$ below) and the Guo-Romano procedure \cite{GR2007} ($\GR$ below). While these two procedures are the most prominent, they were originally designed for homogeneous test statistics, that is, when the null distribution functions of the $p$-values $F_i$, $1\leq i\leq m$, are all equal.  In many applications, however, the data are heterogeneous which leads to heterogeneous null distribution functions. Ignoring this heterogeneity usually  induces a conservativeness for the aforementioned procedures.  In this paper, we develop three new procedures that incorporate the $F_i$'s, while ensuring the FDX control. The heterogeneous version of $\LR$, denoted $\HLR$, is based on the arithmetic average of the $F_i$'s, while the heterogeneous version of $\GR$, denoted $\HGR$, is based on the geometric average of the $F_i$'s. We also introduce a  procedure $\PB$, that is based on the Poisson-binomial distribution and that uniformly improves $\HLR$ and $\HGR$,  at the price of a higher computational complexity. Perhaps surprisingly, this shows that, contrary to the known theory of false discovery rate (FDR) control under heterogeneity, the way to incorporate the $F_i$'s can be particularly simple  in the case of FDX control, and does not require any further correction term. 
The performances of the new proposed procedures are illustrated by real and simulated data in two important heterogeneous settings: first, when the test statistics are continuous but the $p$-values are weighted by some known independent weight vector, e.g., coming from co-data sets;  second,  when the test statistics are discretely distributed, as is the case for data representing frequencies or counts.
\end{abstract}

\maketitle

\section{Introduction}

\subsection{Background}

 When many statistical tests are performed simultaneously, 
a ubiquitous way to account for the erroneous rejections of the procedure is the false discovery proportion (FDP), that is, the proportion of errors in the rejected sets, as introduced in the seminal paper \cite{BenjaminiHochberg95}. Most of the related literature studies the expected value of this quantity, which is the false discovery rate (FDR), e.g., building procedures that improve the original Benjamini-Hochberg procedure by trying to adapt to some underlying structure of the data. In particular, a fruitful direction is to take into account the heterogeneous structure of the different tests.  Heterogeneity may originate from various sources. The two main examples we have in mind, and which have been intensively investigated in the statistical literature recently, is heterogeneity caused by $p$-value weighting and discrete data.

The $p$-value weighting is a popular approach that can be traced back to  \cite{Holm1979} and that has been further developed specifically for FDR in, e.g.,  \cite{GRW2006,BR2008EJS,HZZ2010,ZZ2014,Ramdas2017}. 
Here, the heterogeneity can be for instance driven by sample size, groups, or more generally by some covariates. 
In particular, finding 
optimal weighting in the sense of maximizing the number of true rejections has been investigated in \cite{WR2006,RDV2006,RW2009,Ign2016,Durand2017}, either from independent weighting or from the same data set. As a result, the weighted $p$-values have heterogeneous null distribution functions $\{F_i,1\leq i\leq m\}$ that must be properly taken into account by multiple testing procedures.  

On the other hand,
multiple testing for discrete distributions is a well identified research field \cite{Tar1990,WestWolf1997,Gilbert05} that has received a growing attention  in the last decade, see, e.g.,
\cite{Heyse2011,Heller2012,Dickhaus2012,Habiger2015,CDS2015,Doehler2016,CDH2018,DDR2018,DDR2019} and references therein. 
The most typical setting is the case for which each test is performed according to a contingency table. In that situation, the heterogeneity is induced by the fact that marginal counts naturally vary from one table to another. The approach is then to suitably combine the heterogeneous null distributions to compensate the natural conservativeness of individual discrete tests.
Namely, 
Heyse's approach \cite{Heyse2011} 
 is to 
consider the transformation
\begin{equation}\label{equFbar}
\ol{F}(t) = m^{-1}\sum_{i=1}^m F_i(t),\:\:t\in[0,1],
\end{equation}
and to apply BH to the transformed $p$-values $\{\ol{F}(p_i),1\leq i\leq m\}$. Unfortunately, the latter does not rigorously control the FDR, as it has been proven in \cite{Doehler2016,DDR2018}. Appropriate corrections of the $\ol{F}$ expression have been  proposed in \cite{DDR2018} in order to recover rigorous FDR control.

\subsection{FDX control}
A common criticism of FDR is that it 
captures only  the average behavior of the FDP. In particular,  controlling the FDR does not prevent the FDP from possessing undesirable fluctuations and we may aim to stochastically control of FDP in other ways.
A classical approach is to control the false probability exceedance (FDX) in the following sense: for $\alpha, \zeta \in (0,1)$, 
\begin{equation}\label{equFDX}
\FDX=\P(\FDP>\alpha) \leq \zeta.
\end{equation}
This corresponds to control the $(1-\zeta)$-quantile of the FDP distribution at level $\alpha$, see, e.g., \cite{GW2004,PGVW2004,Korn2004,LR2005,GW2006,RW2007,GHS2014,DR2015}.
Let us also mention that the probabilistic fluctuation of the FDP process is  of interest in its own, see, e.g., \cite{Neu2008,RV2011,DR2011,DR2015a,DJ2019}.

Among multiple testing procedures,  step-down procedures have been shown to be particularly useful for FDX control. 
Two prominent step-down procedures have been proven to control the FDX under different distributional assumptions:
\begin{itemize}
\item The Lehmann-Romano procedure $\LR$, introduced in \cite{LR2005},  is defined as the step-down procedure with critical values 
\begin{align}
\tau^{\mbox{{\tiny LR}}}_\l=\zeta\frac{\lfloor \alpha \l\rfloor+1}{m(\l)}\:\:, 1\leq \l\leq m
\label{crit:LR},
\end{align}
where we denote 
\begin{equation}\label{equml}
m(\l)=m-\ell + \lfloor \alpha \l\rfloor+1.
\end{equation}
 It has been shown to control the FDX  under various dependence assumptions between the $p$-values, e.g., when each $p$-value under the null is independent of the family of the $p$-values under the alternative (Theorem~3.1 in \cite{LR2005}), which we will refer to \eqref{Indep0} below, or when the Simes inequality holds true among the family of true null $p$-values (Theorem~3.2 in \cite{LR2005}). 
\footnote{Under the latter condition, it has also been proven later that 
the step-up version of $\LR$, that is, the step-up procedure using the critical values \eqref{crit:LR} also controls the FDX, see the proof of Theorem~3.1 in \cite{GHS2014}. }

\item The procedure $\LR$ has been improved by the Guo-Romano procedure $\GR$, see \cite{GR2007}, defined as the step-down procedure with critical values
\begin{align}
\tau^{\mbox{{\tiny GR}}}_\l&= \max\{t\in [0,1]\::\: \P( \Bin[m(\l), t]  \geq \lfloor \alpha \l\rfloor+1) \leq \zeta\},\:\: 1\leq \l\leq m,
\label{crit:GR}
\end{align}
where $\Bin[n,p]$ denotes any variable following  a binomial distribution with parameters $n$ and $p$. 
While making more rejections, the procedure $\GR$ controls the FDX under a stronger assumption: the null $p$-value family and the alternative $p$-value family are independent and that the null $p$-values are independent, which we will refer to \eqref{Indep} below. 

\end{itemize}

\subsection{Contributions}

The global aim of the paper is to 
improve procedures $\LR$ and $\GR$ by incorporating the null distribution functions $\{F_i,1\leq i\leq m\}$ of the $p$-values while maintaining  rigorous FDX control.
More specifically, the contributions of this work are as follows:
\begin{itemize}
\item[\textbullet]
we introduce the heterogeneous Lehmann Romano procedure $\HLR$, which controls the FDX under \eqref{Indep0} and is an uniform improvement of $\LR$ (when the marginals of the null $p$-values are super-uniform, see \eqref{superunif} further on); 
\item[\textbullet] 
we introduce the heterogeneous Guo Romano procedure $\HGR$, 
 which controls the FDX under \eqref{Indep} and is an uniform improvement of $\GR$ (under \eqref{superunif});
\item[\textbullet] 
at the price of additional computational complexity, we introduce the Poisson-binomial procedure $\PB$, which controls the FDX under \eqref{Indep} and is a uniform improvement of $\HLR$ and $\HGR$;
\item[\textbullet]  we apply this new technology to weighted $p$-values to provide the first weighted procedures that control the FDX (to our knowledge), called $\WLR$ and $\WGR$. They are able to improve their non-weighted counterparts $\LR$ and $\GR$, respectively, see Section~\ref{sec:weighting};
\item[\textbullet] in the discrete context, our new procedures re-named $\DLR$, $\DGR$ are shown to be uniform improvements with respect to the continuous procedures $\LR$ and $\GR$, respectively. To the best of our knowledge, these are the first FDX controlling procedures tailored specifically to discrete $p$-value distributions. The amplitude of the improvement can be substantial, as we show both with simulated and real data examples, see Section~\ref{sec:discrete}. 
\end{itemize}
The paper is organized as follows: Section~\ref{sec:framework} introduces the statistical setting, the procedures and FDX criterion, as well as a shortcut to compute our step-down procedures without evaluating the critical values. Section~\ref{sec:NewFDXProcedures} is the main section of the paper, which introduces the new heterogeneous procedures and their FDX controlling properties. Our methodology is then applied in two particular frameworks: new weighted procedures controlling the FDX are derived in Section~\ref{sec:weighting} while 
 Section~\ref{sec:discrete} is devoted to the case where the tests are discrete. Both sections include numerical illustrations. A discussion is provided in Section~\ref{sec:discussion} and most of technical details are deferred to Section~\ref{sec:theory}. Appendix~\ref{sec:appendix} gives additional numerical details for simulations.

\section{Framework}\label{sec:framework}

\subsection{Setting}\label{sec:setting}

We use here a classical formal setting for heterogeneous nulls, see, e.g.,  \cite{DDR2018}. We observe $X$, defined on an abstract probabilistic
space, valued in an observation space
$(\mathcal{X},\mathfrak{X})$ and of distribution  $P$ that 
belongs to a set $\mathcal{P}$ of
possible distributions. We consider $m$ null hypotheses for $P$, denoted $H_{0,i}$ ,
$1\leq i \leq m$,  and we denote the
corresponding set of true null hypotheses by $\cH_0(P)=\{1\leq i
\leq m\::\: \mbox{$H_{0,i}$ is satisfied by $P$}\}$. We also
denote by $\cH_1(P)$ the complement of $\cH_0(P)$ in
$\{1,\dots,m\}$ and by $m_0(P)=|\cH_0(P)|$ the number of true
nulls.

We assume that there exists a set of $p$-values that is, a set of random variables $\{p_i(X), 1\leq i
\leq m\}$, valued in $[0,1]$. We introduce the following dependence assumptions between the $p$-values:
\begin{align}
&\mbox{for all $P\in \mathcal{P}$, $\{p_i(X),  i \in \cH_0(P)\}$ is independent of $\{p_i(X),  i \in \cH_1(P)\}$;}\label{Indep0}\tag{Indep0}\\
&\mbox{\eqref{Indep0} holds and for all $P\in \mathcal{P}$, $\{p_i(X),  i \in \cH_0(P)\}$ consists of independent variables.}\label{Indep}\tag{Indep}
\end{align}
Note that \eqref{Indep0} and \eqref{Indep} are both satisfied when all the  $p$-values  $p_i(X)$, $1\leq i \leq m$, are mutually independent in the model $\mathcal{P}$. 
The (maximum) null cumulative distribution function of each $p$-value is denoted 
\begin{equation}\label{equ:Fi}
 F_{i}(t) = \sup_{P\in
\mathcal{P}\::\: i\in \cH_0(P)} \{\P_{X\sim P}(p_{i}(X)\leq t )\}, \:\:
t\in[0,1], \:\:1\leq i\leq m.
\end{equation}
We let $\mathcal{F}=\{F_{i}, 1\leq i \leq m\}$ that we supposed to be {\it known} and we consider the following possible situations for the functions in $\mathcal{F}$:
\begin{align}
&\mbox{for all $i \in\{1,\dots,m\}$, $F_i$ is continuous on $[0,1]$} \tag{Cont} \label{cont}\\
&\begin{array}{c}\mbox{for all $i \in\{1,\dots,m\}$, there exists some countable set $\mathcal{A}_i\subset [0,1]$ such that}\\
\mbox{$F_i$ is a step function, right continuous, that jumps only at some points of $\mathcal{A}_i$.}
\end{array}
 \tag{Discrete} \label{discrete}
\end{align}
The case \eqref{discrete} typically arises when for all $P\in \mathcal{P}$ and $i\in\{1,\dots,m\}$, $\P_{X\sim P}(p_i(X)\in \mathcal{A}_i)=1$.
Throughout the paper, we will assume that we are either in the case \eqref{cont} or \eqref{discrete} and we denote $\mathcal{A}=\cup_{i=1}^m \mathcal{A}_i$, with by convention $\mathcal{A}_i=[0,1]$ when \eqref{cont} holds.
For comparison with the homogeneous case, we also let
\begin{align}
&\mbox{for all $i \in\{1,\dots,m\}$,  $F_{i}(t)\leq t$ for all $t\in[0,1]$}.\tag{SuperUnif}\label{superunif}
\end{align}

\subsection{False Discovery Exceedance and step-down procedures}

In general, a multiple testing procedure is defined as a random subset  $R=R(X)\subset \{1,\dots,m\}$ which corresponds to the indices of the rejected nulls.
For $\alpha\in(0,1)$, the false discovery exceedance of $R$ is defined as follows:
\begin{align}\label{equ:FDX}
\FDX_{\alpha}(R,P)  = \P_{X\sim P}\left( \frac{|R(X)\cap \cH_0(P)|}{|R(X)|\vee 1} > \alpha \right),\:\:\:P\in\mathcal{P}.
\end{align}

In this paper, we consider particular multiple testing procedures, called step-down procedures. Given some $p$-value family $(p_i)_{1\leq i \leq m}$ and some non-decreasing sequence $(\tau_\l)_{1\leq \l\leq m}\in[0,1]^m$, the step-down procedure with critical values $(\tau_\l)_{1\leq \l\leq m}\in[0,1]^m$ rejects the  null hypotheses corresponding to the set  
\begin{align}
R&=\{i\in\{1,\dots,m\}\::\: p_i(X)\leq \tau_{\hat\l}\}\label{SDrejectionsetR}\\
\wh{\l} &= \max\{\l\in \{0,\dots,m\}\::\: \forall \l'\leq \l, \: p_{\sigma(\l')}\leq \tau_{\l'}\}, \:\:\:\mbox{ (convention $p_{\sigma(0)}=0$)},\label{SDrewrite}
\end{align}
for which $
 p_{\sigma(1)}\leq \dots \leq p_{\sigma(m)}
$ denotes the $p$-values $\{p_i(X), 1\leq i \leq m\}$ ordered increasingly (for some data-dependent permutation $\sigma$).

\subsection{Transformation function family and computational shortcut}\label{sec:shortcut}

In this paper, the critical values will be obtained by inverting some functional, that is,  
\begin{equation}\label{inversecritvalues}
\tau_\l = \xi_\l^{-1}(\zeta) = \max\{t\in \mathcal{A}\::\: \xi_\l(t)\leq \zeta\} , \mbox{($\tau_\l=0$ if the set is empty)},\:\:\: 1\leq \l\leq m,
\end{equation}
for $\xi_\l: [0,1]\mapsto [0,\infty)$, $1\leq \l\leq m$, a given set of functions. In order for \eqref{inversecritvalues} to be well-defined and $\l\mapsto \tau_\l $ to be nondecreasing, we  will say that the function set $\{\xi_\l, 1\leq \l\leq m\}$ is a {\it transformation function family} if it satisfies
the following conditions:
\begin{align}
&\begin{array}{l}
\mbox{for all $\l \in\{1,\dots,m\}$, $\xi_\l$ is a non-decreasing function;}\\
\mbox{for all  $t\in[0,1]$ and all $\l\in\{1,\dots,m-1\}$, we have 
$\xi_{\l+1}(t)\leq \xi_\l(t)$;}\\
\mbox{in case \eqref{cont}, for all $\l \in\{1,\dots,m\}$,  $\xi_\l$ is continuous on $[0,1]$.}
\end{array}
 \label{condtransfo}
\end{align}

For instance, the critical values of the procedure $\LR$ can be rewritten as \eqref{inversecritvalues} for the functions
\begin{equation}\label{xilLR}
\xi^{\mbox{{\tiny LR}}}_{\l}(t) =  \frac{m(\l)}{\lfloor \alpha \l\rfloor+1} t,\:\:t\in[0,1],\:\:1\leq \l\leq m.
\end{equation}
We easily  check that the function set $\{\xi^{\mbox{{\tiny LR}}}_\l, 1\leq \l\leq m\}$ is a family of transformation functions (in the sense of \eqref{condtransfo}). Indeed,  $\frac{m-\l+i}{i}$ is non-increasing both in $\l\in\{1,\dots,m\}$ and $i\in\{1,\dots,\lfloor \alpha m\rfloor+1\}$.
A second example is given by the procedure $\GR$ for which 
\begin{align}\label{xi:GR}
\xi^{\mbox{{\tiny GR}}}_\l(t)&=\P( \Bin[m(\l), t] \geq \lfloor \alpha \l\rfloor+1), \:\:\:1\leq \l\leq m,\:t\in[0,1],
\end{align}
can be proved to form a family of transformation functions. Indeed, the only non-obvious argument to prove  \eqref{condtransfo}  is that for a fixed $t\in[0,1]$, and $\l\in\{1,\dots,m-1\}$ we have $\xi^{\mbox{{\tiny GR}}}_{\l+1}(t)\leq \xi^{\mbox{{\tiny GR}}}_\l(t)$. This comes from the fact that 
$
\P( \Bin[m-\l+i, t]  \geq i) =\P( \Bin[m-\l+i, 1-t]  \leq m-\l) 
$ 
is non-increasing both in $i$ and $\l$.

Finally, because of the inversion, computing the critical values via \eqref{inversecritvalues} can be time consuming.  Fortunately, computing the critical values is actually not necessary if we are solely interested in determining the rejection set $R$ given by \eqref{SDrejectionsetR}.  
As the following result shows, we may determine $R$ by working directly with the  transformation functions.

\begin{proposition}\label{prop:equivSD}
Let us consider any transformation function family $\{\xi_\l, 1\leq \l\leq m\}$ 
and the corresponding critical values $\tau_\l$, $1\leq \l\leq m$, defined by \eqref{inversecritvalues}.
Then, for all $P\in\mathcal{P}$, with $P$-probability $1$, the step-down procedure $R$ with critical values $(\tau_\l)_{1\leq \l\leq m}$ can equivalently written as
\begin{align}
R&=\{i\in\{1,\dots,m\}\::\: \tilde{p}_i\leq \zeta\};\label{Rshortcut}\\
\tilde{p}_i& =\max_{\substack{1\leq \l\leq m\\p_{\sigma(\l)}\leq p_i}}\{\xi_\l(p_{\sigma(\l)})\},\:\:1\leq i\leq m.\label{equadjpvalues}
\end{align}
\end{proposition}

Proposition~\ref{prop:equivSD} is proved in Section~\ref{sec:proof:prop:equivSD}.

\section{New FDX controlling procedures} \label{sec:NewFDXProcedures}

In this section, we introduce new procedures $R$ that control the false discovery exceedance at some level $\zeta\in(0,1)$, that is, 
\begin{equation}\label{FDXcontrol}
\mbox{for all $P\in\mathcal{P}$, 
$\FDX_{\alpha}(R,P)\leq \zeta$},
\end{equation}
 while incorporating the family $\{F_i,1\leq i \leq m\}$ in an appropriate way.

\subsection{Tool}

Our main tool is the following bound: For 
any step-down procedure $R$ with critical values $\tau=(\tau_\l)_{1\leq \l\leq m}$, we have
\begin{align}
\sup_{P\in\mathcal{P}}\{\FDX_{\alpha}(R,P)\}&\leq B(\tau,\alpha)\label{toolbound}\\
\mbox{ for }\:\:\:B(\tau,\alpha)&=\sup_{1\leq \l\leq m}\sup_{\substack{P\in\mathcal{P}\\ |\cH_0(P)|\leq m(\l)}} \P_{X\sim P}\left(\sum_{i\in  \cH_0(P)} \mathds{1}\{p_i(X)\leq \tau_\l\} \geq \lfloor\alpha\l\rfloor+1\right).\label{equB}
 \end{align}
 Inequality~\eqref{toolbound}  is valid under the distributional assumption \eqref{Indep0}.  
 This bound comes from a reformulation of Theorem~5.2 in \cite{Roq2011} in our heterogenous framework, see Theorem~\ref{maintool} below. 
Our new procedures
are derived by further upper-bounding $B(\tau,\alpha)$ via various probabilistic devices. More specifically, we will introduce several  transformation function families $\{\xi_\l,1\leq \l\leq m\}$ such that for all $\tau=\{\tau_\l\}_\l$, 
$$
B(\tau,\alpha)\leq \sup_{1\leq \l\leq m} \{\xi_\l(\tau_\l)\}.
$$
According to \eqref{toolbound},  the step-down procedure using the corresponding critical values \eqref{inversecritvalues} will then control the FDX in the sense of \eqref{FDXcontrol}.

\subsection{Heterogeneous Lehmann-Romano procedure}

By using the Markov inequality, we obtain 
\begin{align} \label{eq:BoundLR}
B(\tau,\alpha)&\leq \sup_{1\leq \l\leq m}\sup_{\substack{P\in\mathcal{P}\\ |\cH_0(P)|\leq m(\l)}} \frac{\sum_{i\in  \cH_0(P)} F_i(\tau_\l)}{\lfloor\alpha\l\rfloor+1}= \sup_{1\leq \l\leq m} \frac{\sum_{j=1}^{m(\l)} (F(\tau_\l))_{(j)}}{\lfloor\alpha\l\rfloor+1},
 \end{align}
where $(F(t))_{(1)}\geq \dots \geq (F(t))_{(m)}$ denotes the values of $\{F_i(t),1\leq i\leq m\}$ ordered decreasingly.
Bounding the above quantity by $\zeta$ entails the following procedure.

\begin{definition} 
The heterogeneous Lehmann-Romano procedure, denoted by $\HLR$, is defined as the step-down procedure using the critical values defined by
\begin{align}
\tau^{\mbox{{\tiny HLR}}}_\l & = \max\{t\in \mathcal{A}\::\: \xi^{\mbox{{\tiny HLR}}}_\l(t)\leq \zeta\} ,\:\:\: 1\leq \l\leq m;\label{tau:HLR}\\
\xi^{\mbox{{\tiny HLR}}}_\l(t)&= \frac{\sum_{j=1}^{m(\l)} (F(t))_{(j)}}{\lfloor \alpha \l\rfloor+1}, 
 \:\:\:1\leq \l\leq m,\:t\in[0,1],\label{xi:HLR}
\end{align}
where $(F(t))_{(1)}\geq \dots \geq (F(t))_{(m)}$ denotes the values of $\{F_i(t),1\leq i\leq m\}$ ordered decreasingly and $m(\l)$ is defined by \eqref{equml}.
\end{definition}

The quantity $\xi^{\mbox{{\tiny HLR}}}_\l(t)$ is thus similar to $\xi^{\mbox{{\tiny LR}}}_\l(t)$, in which $t$ has been replaced by 
 the average of the $m(\l)$ largest values of $\{F_i(t),1\leq i\leq m\}$. 
To check that the functions $\xi^{\mbox{{\tiny HLR}}}_\l$
form a transformation function family in the sense of \eqref{condtransfo}, we note that $\frac{1}{m(\l)} \sum_{j=1}^{m(\l)} (F(t))_{(j)}$ is non-increasing in $\l$ (averaging on smaller values makes the average smaller) and continuous in $t$ under \eqref{cont} (because $t\mapsto (F_i(t))_{1\leq i\leq m}$ is continuous and $x\in (\R^m,\|\cdot\|_\infty)\mapsto N^{-1}\sum_{k=1}^N x_{(k)} \in (\R,|\cdot|)$ is $1$-Lipschitz).

In the classical 
case \eqref{superunif}, we have $\xi^{\mbox{{\tiny HLR}}}_\l(t)\leq \xi^{LR}_\l(t)$ for all $t\in[0,1]$ and $1\leq \l\leq m$. Hence,  $\HLR$ is less conservative than $\LR$ in that situation. 
A technical detail is that this only holds almost surely because the range $\mathcal{A}$ in \eqref{tau:HLR} can be different from $[0,1]$ in the case \eqref{discrete}. 
This entails the following result.

\begin{proposition} \label{prop:HLR}
In the setting defined in Section~\ref{sec:setting}, the procedure $\HLR$ satisfies the following
\begin{itemize}
\item[\textbullet] Under \eqref{Indep0}, $\HLR$ controls the $\FDX$ in the sense \eqref{FDXcontrol};
\item[\textbullet] Under \eqref{superunif}, the set of nulls rejected by $\HLR$ contains the one of $\LR$ with $P$-probability $1$, for all $P\in\mathcal{P}$.
\end{itemize}
\end{proposition}

\subsection{Poisson-binomial procedure}\label{sec:HPB}

Here, we propose to bound \eqref{equB} by using the Poisson-binomial distribution. To this end, recall that the Poisson-Binomial distribution of parameters $\pi=(\pi_i)_{1\leq i \leq n}\in[0,1]^n$, denoted $\PBin[\pi]$ below,  corresponds to the distribution of $\sum_{i=1}^n \varepsilon_i$, where the $\varepsilon_i$ are all independent and each $\varepsilon_i$ follows a Bernoulli distribution of parameter $\pi_i$ for $1\leq i\leq n$. 

First note that for all $i\in\cH_0(P)$ and $t\in[0,1]$, we have that $\mathds{1}\{p_i(X)\leq t\}$ is stochastically upper bounded by a Bernoulli variable of parameter $F_i(t)$, see \eqref{equ:Fi}. As a consequence, 
by assuming \eqref{Indep}, we have for all critical values $(\tau_\l)_{1\leq \l\leq m}$,
\begin{align}
B(\tau,\alpha)
&\leq \sup_{1\leq \l\leq m}\sup_{\substack{A\subset \{1,\dots,m\}\\ |A|\leq m(\l)}} \P\left(   \PBin\left[(F_i(\tau_\l))_{i\in A}\right] \geq \lfloor\alpha\l\rfloor+1\right)\nonumber\\
&= \sup_{1\leq \l\leq m} \P\left(   \PBin\left[((F(\tau_\l))_{(j)})_{1\leq j\leq m(\l)}\right] \geq \lfloor\alpha\l\rfloor+1\right).\label{equ:Bpb}
 \end{align}
Bounding the latter by $\zeta$ leads to the following procedure.

\begin{definition} 
The Poisson-binomial procedure, denoted by $\PB$, is defined as the step-down procedure using the critical values 
\begin{align}
\tau^{\mbox{{\tiny PB}}}_\l & = \max\{t\in \mathcal{A}\::\: \xi^{\mbox{{\tiny PB}}}_\l(t)\leq \zeta\} ,\:\:\: 1\leq \l\leq m;\label{tau:PB}\\
\xi^{\mbox{{\tiny PB}}}_\l(t)&=  \P\left( \PBin\left[((F(t))_{(j)})_{1\leq j\leq m(\l)}\right] \geq \lfloor \alpha \l\rfloor+1\right), \:\:\:1\leq \l\leq m,\:t\in[0,1],\label{xi:PB}
\end{align}
where $(F(t))_{(1)}\geq \dots \geq (F(t))_{(m)}$ denotes the values of $\{F_i(t),1\leq i\leq m\}$ ordered decreasingly  and $m(\l)$ is defined by \eqref{equml}.
\end{definition}

Let us now check that $\{\xi^{\mbox{{\tiny PB}}}_\l,1\leq \l\leq m\}$ is a  transformation function family, that is, satisfy \eqref{condtransfo}. The continuity assumption holds because, under \eqref{cont}, the mapping $t\in[0,1]\mapsto ((F(t))_{(j)})_{1\leq j\leq m(\l)}$ is continuous (argument similar to above) and 
the cumulative distribution function of $\PBin[\pi]$ is a continuous function of $\pi\in[0,1]^n$. The monotonic property $\xi^{\mbox{{\tiny HGR}}}_{\l+1}(t)\leq \xi^{\mbox{{\tiny HGR}}}_\l(t)$ comes from the fact that 
$
\P( \PBin\left[((F(t))_{(j)})_{1\leq j\leq m-\l+i}\right]  \geq i) =\P( \PBin\left[(1-(F(t))_{(j)})_{1\leq j\leq m-\l+i}\right]   \leq m-\l) 
$ 
is non-increasing both in $i$ and $\l$.

Since under \eqref{superunif}, the distribution $ \PBin\left[((F(t))_{(j)})_{1\leq j\leq m(\l)}\right] $ is stochastically smaller than the distribution $ \Bin[m(\l), t] $, the following holds.

\begin{proposition} \label{prop:HPB}
In the setting defined in Section~\ref{sec:setting}, the procedure $\PB$ satisfies the following
\begin{itemize}
\item[\textbullet] Under \eqref{Indep}, $\PB$ controls the $\FDX$ in the sense \eqref{FDXcontrol};
\item[\textbullet] Under \eqref{superunif}, the set of nulls rejected by $\PB$ contains the one of $\GR$ with $P$-probability $1$, for all $P\in\mathcal{P}$.
\end{itemize}
\end{proposition}

However, in general, the procedure $\PB$ is computationally demanding, even with the shortcut mentioned in Section~\ref{sec:shortcut}. This comes from the computation of $\xi^{\mbox{{\tiny PB}}}_\l(t)$ which involves the distribution function of a Poisson-binomial variable. 
In the next section, we make $\PB$ slightly more conservative for recovering the computational price of $\GR$.

\subsection{Heterogeneous Guo-Romano procedure}\label{sec:HGR}

In this section, we further upper-bound \eqref{equ:Bpb} by using that any $\PBin\left[(\pi_i)_{1\leq i \leq n}\right]$ random variable is stochastically upper-bounded by 
a $\Bin\left[n, 1- \left(\prod_{i=1}^n (1-\pi_i )\right)^{1/n}\right]$ random variable (see Example 1.A.25 in \cite{Shaked}). This yields
\begin{align}
B(\tau,\alpha)
&\leq \sup_{1\leq \l\leq m}  \P\left(   \Bin\left[m(\l),
\tilde{F}_{m(\l)}(\tau_\l)
 \right] \geq \lfloor\alpha\l\rfloor+1\right) ,
\label{equ:Bpb}
 \end{align}
where we let 
\begin{align}\label{equFjt}
\tilde{F}_{j}(t) =1-\left(\prod_{j'=1}^{j} (1-(F(t))_{(j')})\right)^{1/j},  \:\:\:1\leq j\leq m,\:t\in[0,1],
 \end{align}
where $(F(t))_{(1)}\geq \dots \geq (F(t))_{(m)}$ denotes the values of $\{F_i(t),1\leq i\leq m\}$ ordered decreasingly.

This reasoning suggests another heterogeneous procedure, based on the binomial distribution.
Since $\GR$ also uses the binomial device, we name this new procedure  the heterogeneous Guo-Romano procedure.

\begin{definition} 
The heterogeneous Guo-Romano procedure, denoted by $\HGR$, is defined as the step-down procedure using the critical values defined by
\begin{align}
\tau^{\mbox{{\tiny HGR}}}_\l & = \max\{t\in \mathcal{A}\::\: \xi^{\mbox{{\tiny HGR}}}_\l(t)\leq \zeta\} ,\:\:\: 1\leq \l\leq m;\label{tau:HGR}\\
\xi^{\mbox{{\tiny HGR}}}_\l(t)&=  \P\left(  \Bin\left[m(\l),
\tilde{F}_{m(\l)}(t) \right] \geq \lfloor \alpha \l\rfloor+1\right), \:\:\:1\leq \l\leq m,\:t\in[0,1],\label{xi:HGR}
\end{align}
where $\tilde{F}_{j}(t)$ is defined in \eqref{equFjt}  and $m(\l)$ is defined by \eqref{equml}.
\end{definition}

The condition \eqref{condtransfo} also holds in that case. However, the proof of monotonicity of $\xi^{\mbox{{\tiny HGR}}}_\l(t)$ is slightly more involved than above and is deferred to Lemma~\ref{lem:monotoneHGRstar}. 
In addition, since  under \eqref{superunif} we have $\tilde{F}_{m(\l)}(t) \leq t$, we deduce that $\HGR$, although more conservative than $\PB$, is still a uniform improvement over $\GR$. 

\begin{proposition} \label{prop:HGR}
In the setting defined in Section~\ref{sec:setting}, the procedure $\HGR$ satisfies the following
\begin{itemize}
\item[\textbullet] Under \eqref{Indep}, $\HGR$ controls the $\FDX$ in the sense \eqref{FDXcontrol};
\item[\textbullet] Under \eqref{superunif}, the set of nulls rejected by $\HGR$ contains the one of $\GR$ with $P$-probability $1$, for all $P\in\mathcal{P}$.
\end{itemize}
\end{proposition}

\begin{remark}
The numerical results in Sections \ref{sec:weighting} and \ref{sec:discrete} suggest that the conservatism of $\HGR$ with respect to $\PB$ is usually quite small. In addition, since the computational effort required by $\HGR$ is comparable to that of $\GR$, the gain in efficiency may be great, especially for large $m$. We therefore think that $\HGR$ may be especially useful for very high dimensional heterogeneous data.
\end{remark}

\begin{remark}
We can also define a non-adaptive version of $\HGR$, defined as the step-down procedure of critical values \eqref{inversecritvalues} based on the transformation functional
\begin{align*}
\xi_\l(t)&=  \P\left(  \Bin\left[m,
\tilde{F}(t) \right] \geq \lfloor \alpha \l\rfloor+1\right), \:\:\:1\leq \l\leq m,\:t\in[0,1],
\end{align*}
where $\tilde{F}(t) =1-\left(\prod_{j=1}^{m} (1-F_j(t))\right)^{1/m}$. While being more conservative than $\HGR$, it still controls the $\FDX$ in the sense \eqref{FDXcontrol} .
So, while controlling the FDR is linked to the arithmetic average of the $F_i$'s  (Heyse's procedure, see text below \eqref{equFbar}), this shows that controlling the FDX is linked to  geometric averaging. In addition, this shows that the situation is even more simple for FDX, because no further correction is needed here, whereas the arithmetic average should be slightly modified in order to yield a rigorous FDR control \cite{DDR2018}.
\end{remark}

\section{Application to weighting}\label{sec:weighting}

It is well known that $p$-value weighting can improve the power of multiple testing procedures, see, e.g., \cite{GRW2006,RW2009,Ign2016,Dur2017,Ramdas2017} and references therein. However, to the best of our our knowledge, except for the augmentation approach described in \cite{GRW2006}, no methods are available that incorporate weighting for FDX control.  We show in this section that such methods can be obtained directly from the bounds on $B(\tau, \alpha)$ introduced in Section~\ref{sec:NewFDXProcedures}.

Throughout this section, we assume that we have at hand a $p$-value family satisfying \eqref{cont} and \eqref{superunif}. As explained in our introduction section (see references therein), while the null distributions of the $p$-values are  typically uniform, the point is that they can have heterogeneous alternative distributions, so that it could be desirable to weigh the $p$-values in some way.
For this, we consider a fixed weight vector $(w_i)_{1\leq i \leq m }\in\R_+^m$. The ordered weights are  denoted $w_{(1)}\geq w_{(2)}\geq \dots \geq w_{(m)}$, the average weight is denoted $\overline{w}=m^{-1}\sum_{i=1}^m w_i$ and the average over the $j$ largest weights is denoted by $\overline{w}_j=j^{-1}\sum_{j'=1}^{j} w_{(j')}$.

Since the heterogeneous procedures $\HLR$, $\PB$ and $\HGR$ introduced in Section~\ref{sec:NewFDXProcedures} yield valid control for any collection of distribution functions $\{F_i,1\leq i\leq m\}$, it is possible to use very flexible weighting schemes. In order to limit the scope of this paper, we consider only two simple types of weighting approaches in more detail:
\begin{itemize}
	\item for \emph{arithmetic mean weighting} (abbreviated in what follows as AM), we define the weighted $p$-value family as
	\begin{equation}\label{weightedpvaluesLR}
	p^w_i
	=p_i \:\bar w/w_i, \:\:\:1\leq i\leq m.
	\end{equation}
	The weighted $p$-values thus have  the heterogeneous distribution functions 
	\begin{align}
	\FAM_i(t) &= \left(\frac{w_i}{\bar w} t\right) \wedge 1, \:\:\:1\leq i\leq m,
	\end{align}
	under the null. This corresponds to classical weighting approaches established for FWER and FDR control.
	\item 
	for \emph{geometric mean weighting} (abbreviated as GM), 
	we define 
	\begin{equation}\label{weightedpvaluesGR}
	p^w_i
	=1-(1-p_i)^{\bar w/w_i}, \:\:\:1\leq i\leq m.
	\end{equation}
	The weighted $p$-values therefore have  the following heterogeneous distribution functions under the null:
	\begin{align}
	\FGM_i(t) &= 1-(1-t)^{w_i/\bar w}, \:\:\:1\leq i\leq m.
	\end{align}
\end{itemize}
Thus, combining these two weighting approaches with the three heterogeneous procedures introduced in the previous section yields a total of six weighted procedures which we discuss in more detail below. Note that  a Taylor expansion provides $\FAM_i(t)  \approx \FGM_i(t)$ for small values of $t$. Therefore, we expect that AM and GM procedures will yield similar rejection sets for small $p$-values.

\subsection{Weighted Lehmann-Romano procedures}

Using \eqref{toolbound} and \eqref{eq:BoundLR} (Markov device), we get that any step-down procedure using the weighted $p$-values \eqref{weightedpvaluesLR} and critical values $(\tau_i)_{1\leq i \leq m}$ has a FDX smaller than or equal to
\begin{align*}
\frac{1}{\lfloor \alpha \l\rfloor+1} \sum_{j=1}^{m(\l)} (\FAM(t))_{(j)}
= \frac{1}{\lfloor \alpha \l\rfloor+1} \sum_{j=1}^{m(\l)} \left( \left(\frac{w_{(j)}}{\bar w} t \right)\wedge 1 \right)
\leq \frac{m(\l)}{\lfloor \alpha \l\rfloor+1} \times \frac{\overline{w}_{m(\l)}}{\bar w} t =: \xi^{\mbox{{\tiny wLR-AM}}}_\l(t).
\end{align*}
Since the $\xi^{\mbox{{\tiny wLR-AM}}}_\l$ form a transformation function family, bounding the latter by $\zeta$ leads to an FDX controlling procedure, that we call the AM-weighted Lehmann-Romano procedure, denoted by $\WLRAM$ in the sequel. It thus corresponds to the step-down procedure using the weighted $p$-values \eqref{weightedpvaluesLR} and the critical values 
\begin{align*}
\tau^{\mbox{{\tiny wLR-AM}}}_\l & = \zeta \frac{\lfloor\alpha\l\rfloor+1}{\sum_{j=1}^{m(\l)} w_{(j)}} \bar w=\tau^{\mbox{{\tiny LR}}}_\l \cdot \frac{\bar w}{\overline{w}_{m(\l)}}  , \qquad \:\: 1\leq \l\leq m. 
\end{align*}
In particular, if the weight vector is uniform, that is, $w_i=1$ for all $i$, then  $\WLRAM$ reduces to $\LR$. 

Similarly to above, using the GM weighting \eqref{weightedpvaluesGR} gives an FDX smaller than or equal to
\begin{align*}
	 \frac{1}{\lfloor \alpha \l\rfloor+1}\sum_{j=1}^{m(\l)} (\FGM(t))_{(j)} = \frac{1}{\lfloor \alpha \l\rfloor+1} \sum_{j=1}^{m(\l)} (1-(1-t)^{w_{(j)}/\bar w} )=:\xi^{\mbox{{\tiny wLR-GM}}}_\l(t) .
	\end{align*}
This gives rise to the GM-weighted Lehmann-Romano procedure, denoted $\WLRGM$, defined as the step-down procedure using the weighted $p$-values \eqref{weightedpvaluesGR} and the critical values 
\begin{align*}
\tau^{\mbox{{\tiny wLR-GM}}}_\l & = \max \{ t \in [0,1] : \xi^{\mbox{{\tiny wLR-GM}}}_\l(t)  \le \zeta \} , \:\: 1\leq \l\leq m.
	\end{align*}

In general, no domination relationship holds between $\WLRGM$ and $\WLRAM$.
Finally, again, in case of uniform weighting,   $\WLRGM$  reduces to $\LR$.

\subsection{Weighted Poisson-binomial procedures} \label{sec:WPB}
Applying the strategy of Section~\ref{sec:HPB} with the c.d.f. sets $\{\FAM_i,1\leq i\leq m\}$ and $\{\FGM_i,1\leq i\leq m\}$, we can use the two transformation function families given by
\begin{align*}
\xi^{\mbox{{\tiny wPB-AM}}}_\l(t)&=  \P\left( \PBin\left[\left(\left(\frac{w_{(j)}}{\bar w} t\right) \wedge 1\right)_{1\leq j\leq m(\l)}\right] \geq \lfloor \alpha \l\rfloor+1\right) , \:\:\:1\leq \l\leq m,\:t\in[0,1];\\
\xi^{\mbox{{\tiny wPB-GM}}}_\l(t)&=  \P\left( \PBin\left[(1-(1-t)^{w_{(j)}/\bar w}
)_{1\leq j\leq m(\l)}\right]\geq \lfloor \alpha \l\rfloor+1\right), \:\:\:1\leq \l\leq m,\:t\in[0,1],
\end{align*}
to define new step-down procedures, denoted $\WPBAM$ and $\WPBGM$ respectively, that both ensure FDX control.

\subsection{Weighted Guo-Romano procedures} \label{ssec:WeightedGuoRomano}

We apply here the strategy of Section~\ref{sec:HGR} for the c.d.f. sets $\{\FAM_i,1\leq i\leq m\}$ and $\{\FGM_i,1\leq i\leq m\}$. According to \eqref{equFjt}, let us define 
\begin{align*}
\tilde{F}^{\textnormal{AM}}_{j}(t) &=1-\left(\prod_{j'=1}^{j} \left((1-\left(\frac{w_{(j')}}{\bar w} t\right) \wedge 1 \right) \right)^{1/j},  \:\:\:1\leq j\leq m,\:t\in[0,1];\\
\tilde{F}^{\textnormal{GM}}_{j}(t) &=1-\left(\prod_{j'=1}^{j} (1-t)^{w_{(j')}/\bar w}
\right)^{1/j} = 1-(1-t)^{\overline{w}_{j} /\bar w
},  \:\:\:1\leq j\leq m,\:t\in[0,1].
\end{align*}
This gives rise to the transformation function families 
\begin{align*}
\xi^{\mbox{{\tiny wGR-AM}}}_\l(t)&=  \P\left(\Bin\left[m(\l),
\tilde{F}^{\textnormal{AM}}_{m(\l)}(t)\right] \geq \lfloor\alpha\l\rfloor+1\right) , \:\:\:1\leq \l\leq m,\:t\in[0,1];\\
\xi^{\mbox{{\tiny wGR-GM}}}_\l(t)&= \P\left(\Bin\left[m(\l),
\tilde{F}^{\textnormal{GM}}_{m(\l)}(t)\right] \geq \lfloor\alpha\l\rfloor+1\right), \:\:\:1\leq \l\leq m,\:t\in[0,1].
\end{align*}
Critical values $\tau^{\mbox{{\tiny wGR-AM}}}$ and $\tau^{\mbox{{\tiny wGR-GM}}}$  are obtained via \eqref{inversecritvalues} from families $\xi^{\mbox{{\tiny wGR-AM}}}$ and $\xi^{\mbox{{\tiny wGR-GM}}}$, respectively. This yields two new FDX controlling step-down procedures that are denoted by $\WGRAM$ and $\WGRGM$, respectively. 
Note that, similar to arithmetic weighting for the $\LR$ procedure, geometric weighting leads to a simple transformation of the original $\GR$ critical values, given by
\begin{align}\label{CriticalValues:WGM:GR}
\tau^{\mbox{{\tiny wGR-GM}}}_\l &= 1 - \left( 1- \tau^{\mbox{{\tiny GR}}}_\l \right)^{\bar w / \overline{w}_{m(\l)} }.
\end{align}
Thus, this particular procedure combines simplicity with a close relationship to the original Guo-Romano procedure. By contrast, as for the heterogeneous version, the weighted Poisson-binomial procedures require the evaluation of the Poisson-binomial distribution function which may be computationally demanding for large $m$. The weighted Guo-Romano procedures, on the other hand, while possibly sacrificing some power, only require evaluation of the standard binomial distribution.

\subsection{Analysis of RNA-Seq data}

We revisit an analysis of the RNA-Seq data set 'airway'  using results from the independent hypothesis weighting (IHW) approach (for details, see \cite{Ign2016} and the vignette accompanying its software implementation). Loosely speaking, this method aims to increase power by assigning a weight $w_i$ to each hypothesis and subsequently applying e.g. the Bonferroni or the Benjamini-Hochberg procedure $\BH$ to the weighted $p$-values while aiming for control of $\FWER$ or $\FDR$.

In what follows, we present some results for weighted $\FDX$ control, using the procedures introduced in Sections~\ref{sec:WPB} and~\ref{ssec:WeightedGuoRomano}. For this data set we have $m=64102$ and the weights $w_1, \ldots , w_m$ are taken from the output of the ihw function from the bioconductor package 'IHW'. For the sake of illustration  we assume the $p$-values to be independent. A large portion (about $45\%$) of these weights are 0, Figure \ref{fig:PlotIHWAnalysisWeights} presents a histogram of the (strictly) positive weights.

\begin{figure}[htbp]
	\centering
	\includegraphics[width=1.0\textwidth]{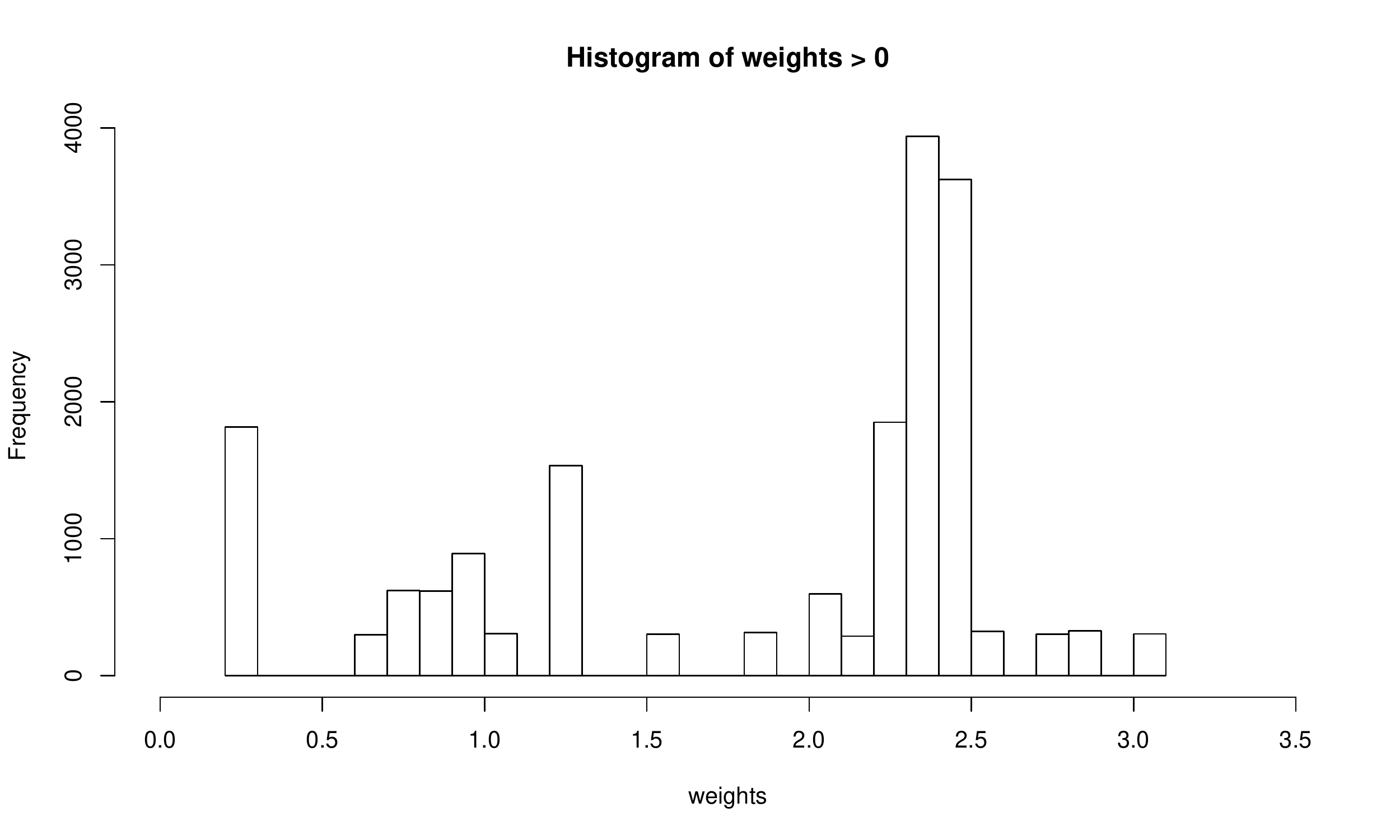}
	\caption{Histogram of positive weights generated by the ihw function for the airway data}
	\label{fig:PlotIHWAnalysisWeights}
\end{figure}

\begin{figure}[htbp]
	\centering
	\makebox{\includegraphics[width=1.0\textwidth]{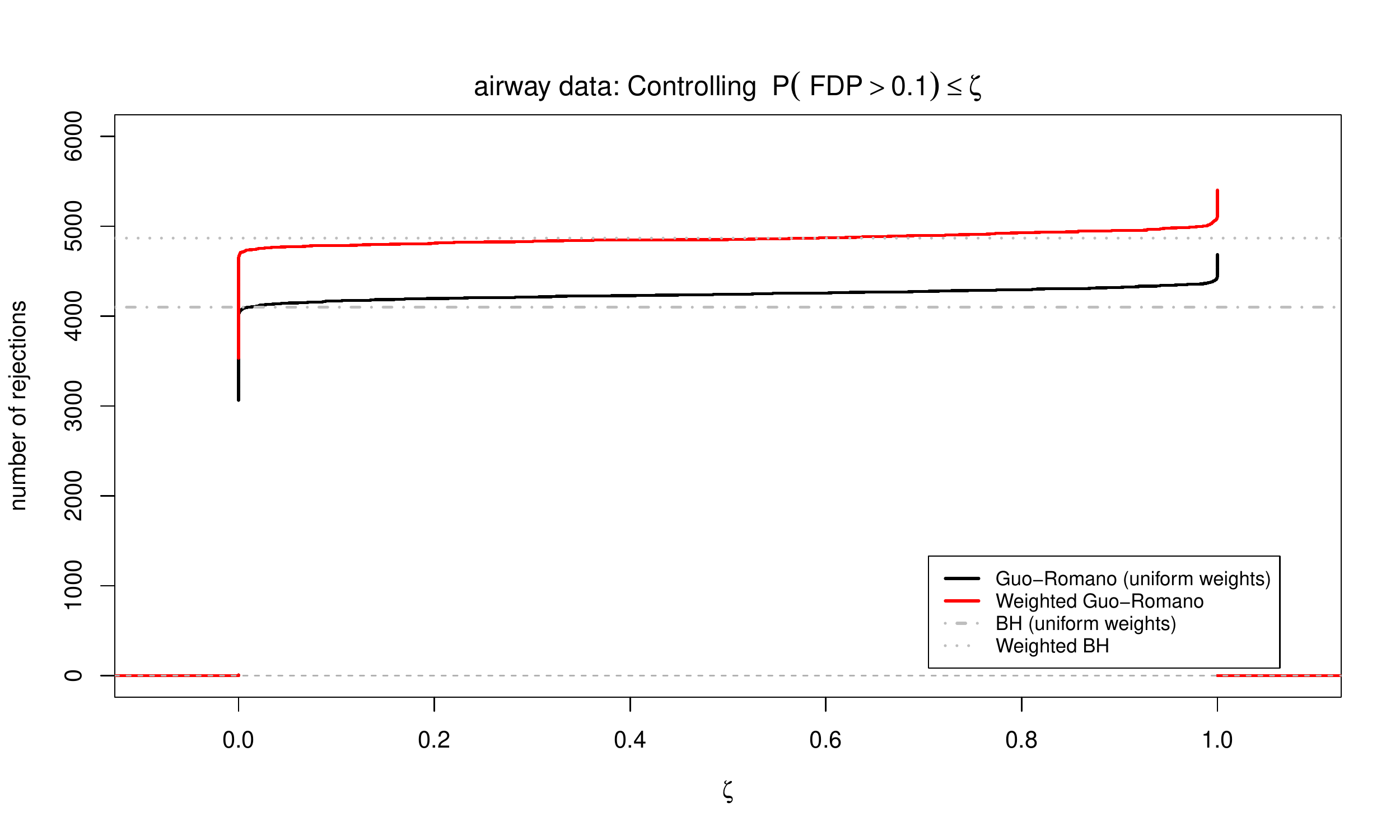}}
	\caption{Number of rejections ($y$-axis) for the airway data when using the $\GR$ and [wGR-GM] procedure. Both procedures control the tail probabilities (on the $x$-axis) for the $\FDP$ exceeding $10\%$. The horizontal lines represent the rejections of the BH and weighted BH procedures at $\FDR$-level $10\%$. }
	\label{fig:PlotIHWAnalysisFDXRejections}
\end{figure} 

Table \ref{tab:AirwayData} shows that controlling the mean (i.e. $\FDR$) or the median of the $\FDP$ leads to similar number of rejections.  
\begin{table}[htbp]
	\caption{Number of rejections  for the airway data. The FDR procedures control $\FDR$ at level $10\%$, the FDX procedures control $P(\FDP > 10\%) \le 0.5$. \label{tab:AirwayData}}
	\centering
	\begin{tabular}{lccccccc}
		\toprule
		\multicolumn{1}{c}{} &  $\BH$ 	&  [wBH] & $\GR$  		& $\WPBAM$ & $\WPBGM$ & $\WGRAM$ & $\WGRGM$
		\\
		\midrule
		Rejections  &  4099    & 4896    & 4243    & 4868 & 4865 & 4853 & 4852  \\
		\bottomrule
	\end{tabular}
\end{table}
For both error rates, incorporating weights leads to similar gains in power. For weighted FDX control, the more conservative weighted Guo-Romano procedures exhibit only a slight loss of power with respect to the weighted Poisson-binomial approaches. The difference between arithmetic and geometric weighting is negligible for this data.

Figure \ref{fig:PlotIHWAnalysisFDXRejections} indicates that for the $\FDX$ controlling procedures, the mapping of the confidence level  to the number of rejections is quite flat. This means that statements about the $\FDP$ can be made with high confidence without losing too much power. For instance, requiring that $\FDP \le 10\% $ with confidence at least $95\%$ still allows for 4145 and 4771 rejections using the $\GR$ and $\WGRGM$ procedures.

\section{Application to discrete tests}\label{sec:discrete}\label{sec:numexp}

\subsection{Discrete FDX procedures} Discrete FDX procedures can be defined in a straightforward way by directly using the distribution functions $F_1, \ldots, F_m$ of the discretely distributed $p$-values. The prototypical example we have in mind are multiple conditional tests like Fisher's exact test. In this case, discreteness and heterogeneity arise from conditioning on the observed table margins. We denote the resulting heterogeneous procedures from section \ref{sec:NewFDXProcedures} by $\DLR$ (for $\HLR$), $\DPB$ (for [HPB]) and $\DGR$ (for $\HGR$).

\subsection{Simulation study}\label{sec:simu}

We now investigate the power of the $\DLR$, $\DPB$ and $\DGR$ procedures in a simulation study similar to those described in \citep{Gilbert05}, \citep{Heller2012} and \citep{DDR2018}. We focus on comparing  the performance of the new discrete  procedures  to their continuous counterparts. Since the analysis with $\DPB$ is computationally demanding, we are also interested in investigating the performance of the slightly more conservative, but numerically more efficient  $\DGR$ procedure.  Finally, as above, we also include $\BH$ (Benjamini-Hochberg procedure) as a benchmark.

\subsubsection{Simulated Scenarios}

We simulate a two-sample problem in which a vector of $m$ independent binary responses (``adverse events") is observed for each subject in two groups, where each group consists of $N=25$ subjects. Then, the goal is to simultaneously test the $m$ null hypotheses $H_{0i}:$ ``$p_{1i}=p_{2i}$", $i=1,\ldots,m$, where $p_{1i}$ and $p_{2i}$ are the success probabilities for the $i$th binary response in group 1 and 2, respectively. Before we describe the simulation framework in more detail, we explain how this set-up leads to discrete and heterogeneous $p$-value distributions. Suppose we have simulated two vectors of dimension $m$ where each component represents a count in $ \{0,\ldots,25\}$. This data can be represented by $m$ contingency tables. Now each hypothesis is tested using Fisher's exact test (two-sided) for each contingency table, which is performed by conditioning on the (simulated) pair of marginal counts. Thus, we can determine for every contingency table $i$ the discrete distribution function $F_i$ of the $p$-values for Fisher's exact test under the null hypothesis. For differing (simulated) contingency tables, these induced distributions will generally be heterogeneous and our inference is conditionally on the marginal counts.

We take $m=800,2000$ where $m=m_{1}+m_{2}+m_{3}$ and data are generated so that the response is $Bernoulli(0.01)$ at $m_{1}$ positions for both groups, $Bernoulli(0.10)$ at $m_{2}$ positions for both groups and   $Bernoulli(0.10)$ at $m_{3}$ positions for group 1 and  $Bernoulli(q)$ at $m_{3}$ positions for group 2 where $q=0.15,0.25,0.4$ represents weak,  moderate and strong effects, respectively. 
The null hypothesis is true for the  $m_{1}$ and $m_{2}$ positions while the alternative hypothesis is true for the  $m_{3}$ positions. We also take different configurations for the proportion of false null hypotheses, $m_{3}$ is set to be $10\%$,  $30\%$ and  $80\%$ of the value of $m$, which represents  small, intermediate and large proportion of effects, respectively (the proportion of true nulls $\pi_{0}$ is $0.9$, $0.7$, $0.2$, respectively). Then, $m_{1}$ is set to be $20\%$,  $50\%$ and  $80\%$ of the number of true nulls (that is, $m-m_{3}$) and $m_{2}$ is taken accordingly as $m-m_{1}-m_{3}$. 

For each of the 54 possible parameter configurations specified by $m,m_{3},m_{1}$ and $q$, $10000$ Monte Carlo trials are performed, that is, $10000$ data sets are generated and for each data set, an unadjusted two-sided $p$-value from Fisher's exact test is computed for each of the $m$ positions, and the multiple testing procedures mentioned above are applied at level $\alpha=0.05$. The power of each procedure was estimated as the fraction of the $m_{3}$ false null hypotheses that were rejected, averaged over the $10000$ simulations (TDP, true discovery proportion). 
Note that while our procedures are designed to control the FDP conditionally on the marginal counts, our power results are presented in an unconditional way for the sake of simplicity. 
For random number generation the R-function \textit{rbinom} was used. The two-sided $p$-values from Fisher's exact test were computed using the R-function \textit{fisher.test}. 

\subsubsection{Results}
Table~\ref{tab:addlabel} in Appendix~\ref{sec:appendix} shows that the (average) power of the compared procedures depends primarily on the strength of the signal $q_3 \in \{0.15, 0.25, 0.4\}$. More specifically, Figure \ref{fig:BoxplotsPower} contains some typical plots of the simulation results.

\begin{figure}[htbp]
	\centering
	\makebox{\includegraphics[width=1.0\textwidth]{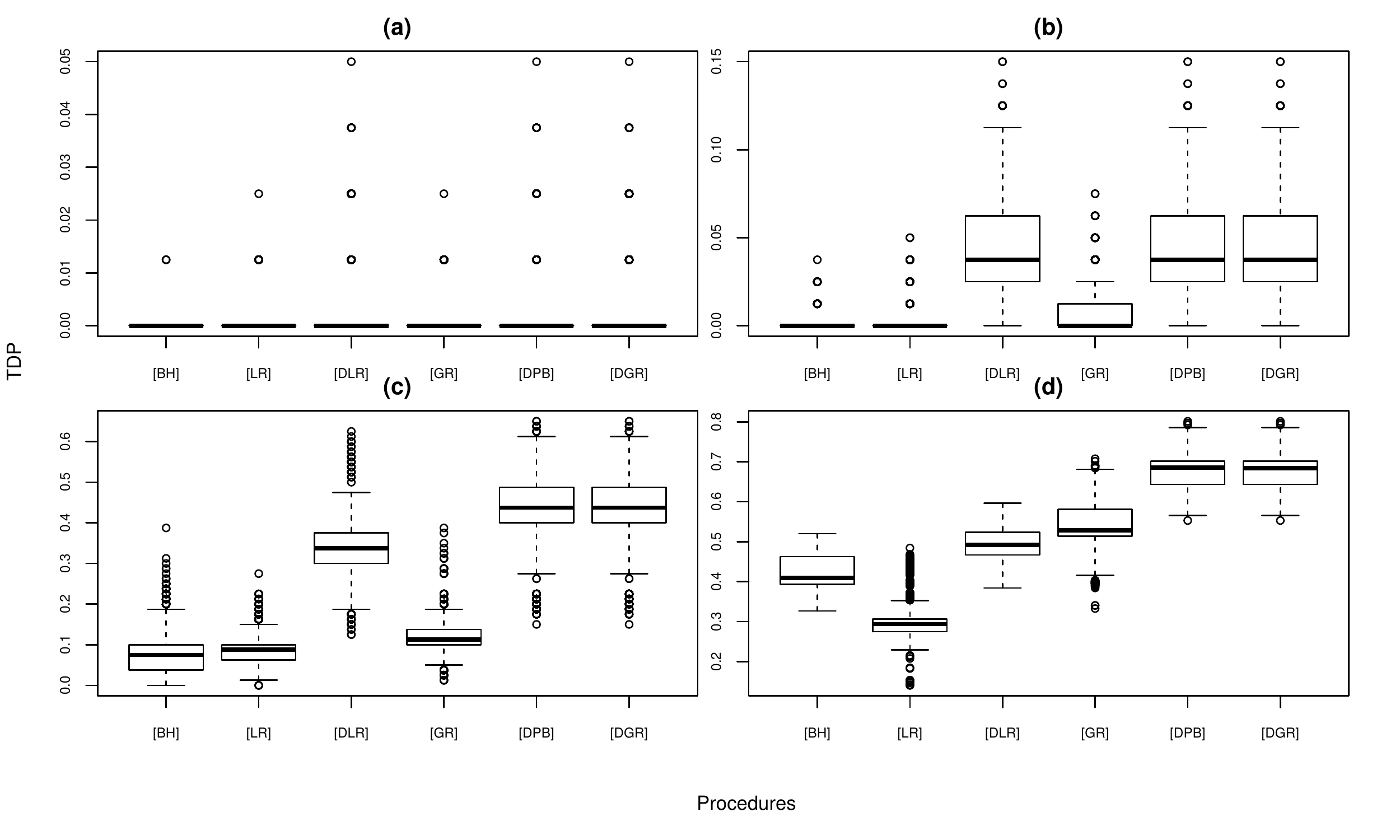}}
	\caption{Boxplots of the simulated true discovery proportions (TDP) for the $\LR$ and $\GR$ procedures, their discrete modifications and the $\BH$ procedure for $m=800$. Panel (a) - (c) show results for $m_3=80$, $m_1=144$ with $q_3=0.15, 0.25, 0.4$, panel (d) shows results for $m_3=640$, $m_1=80$ and $q=0.4$. }
	\label{fig:BoxplotsPower}
\end{figure}  
\begin{itemize}
	\item For $q_3=0.15$, the power of $\BH$, $\LR$ and $\GR$ is practically zero, whereas the discrete procedures are able to reject at least a few hypotheses, see  panel (a) of Figure~\ref{fig:BoxplotsPower}.
	\item For $q_3 = 0.25$, the  power of $\BH$ and $\LR$ stays close to zero, $\GR$  performs slighty better and the discrete variants perform best as illustrated in panel (b) of Figure~\ref{fig:BoxplotsPower}.
	\item For $q_3 = 0.4$, the  power of $\LR$ stays close to zero, while  $\BH$ now rejects a significant amount of hypotheses. The $\DPB$ and $\DGR$ procedures perform  best. If there is a large amount of alternatives, $\GR$ performs better than $\DLR$ (see panel (c) of Figure \ref{fig:BoxplotsPower}) in the other cases, $\GR$ is outperformed by $\DLR$  (see panel (d) of Figure \ref{fig:BoxplotsPower}).
	\item There is no relevant difference in power between the procedures $\DPB$ and $\DGR$.	
\end{itemize}
In summary, these results show that for $\LR$ and $\GR$, significant improvements are possible by using discreteness.

\subsection{Analysis of pharmacovigilance data}
We revisit the analysis of pharmacovigilance data from \cite{Heller2012} presented  in \cite{DDR2018}. This data set is obtained from a database for reporting, investigating and monitoring  adverse drug reactions due to the Medicines and Healthcare products Regulatory Agency in the United Kingdom. It contains the number of reported cases of amnesia as well as the total number of adverse events reported for each of the $m=2446$ drugs in the database. For a more detailed description of the data which is contained in the R-packages \cite{discreteMTP} and \cite{DJDR2018} we refer to \cite{Heller2012}. The works \cite{Heller2012} and \cite{DDR2018} investigate the association between reports of amnesia and suspected drugs by performing for each drug a Fisher's exact test (one-sided) for testing association between the drug and amnesia while adjusting for multiplicity by using several (discrete) FDR procedures. Applying the Benjamini-Hochberg procedure to this data set yields $24$ candidate drugs which could be associated with amnesia. Using the discrete FDR controlling procedures from \cite{DDR2018} yields $27$ candidate drugs.
 
 In what follows, we investigate the performance of the $\LR$, $\DLR$, $\GR$, $\DPB$ and  $\DGR$ procedures for analyzing this data set. First, we compare these procedures when the goal is control of  the median $\FDX$ instead of $\FDR$ at the $5\%$ level, i.e., we require $\P(\FDP > 5\%) \le 0.5$. Figure \ref{fig:FDXPharmacovigilanceCV} illustrates  the data and the critical constants of the  involved $\FDX$ controlling procedures.
 
\begin{figure}[htbp]
	\centering
	\makebox{\includegraphics[width=0.8\linewidth]{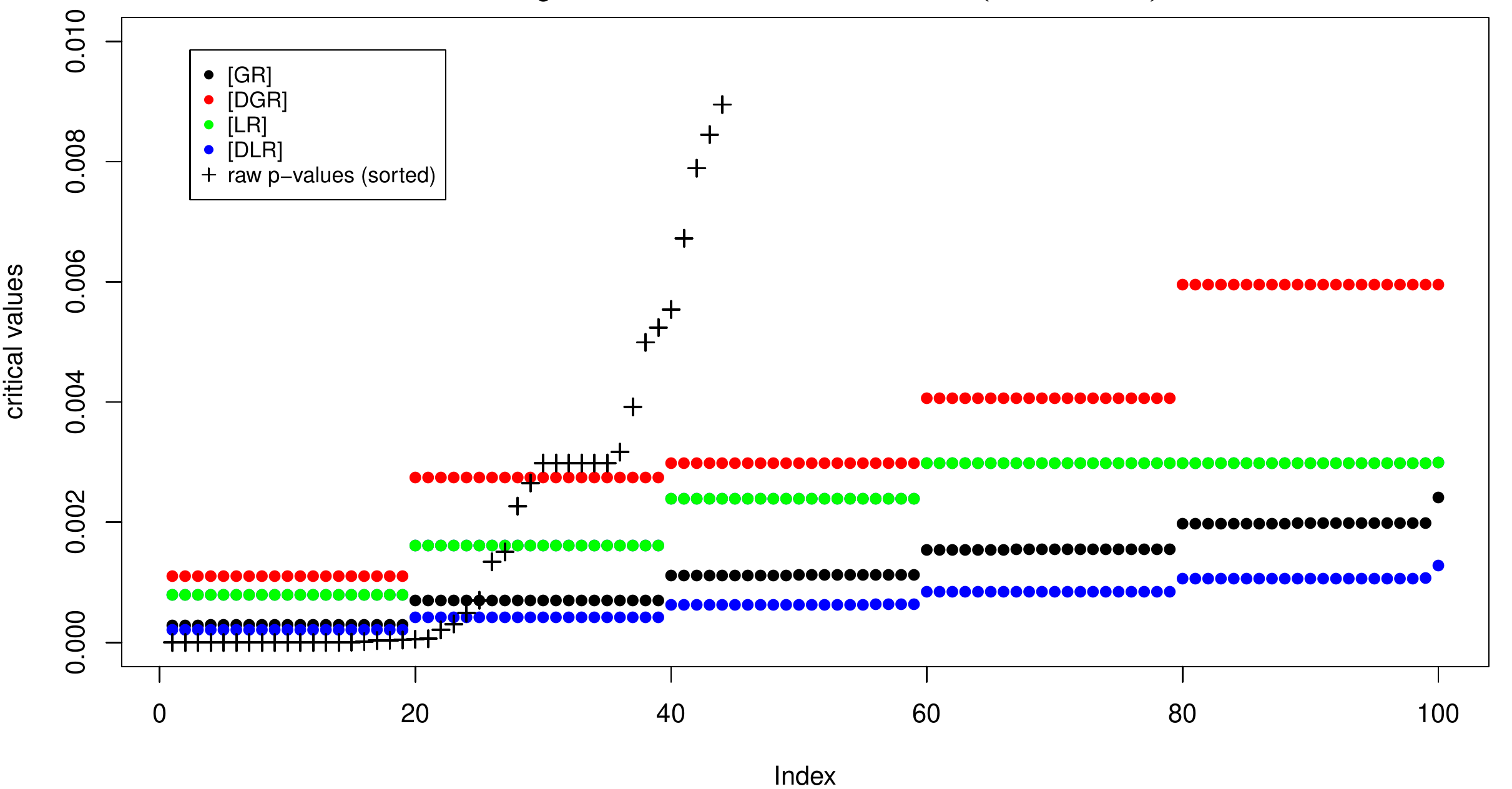}}
	\caption{Critical constants and sorted $p$-values (represented by black crosses) for median FDX control using the $\LR$, $\DLR$, $\GR$ and $\DGR$ procedures for the  pharmacovigilance data. 	\label{fig:FDXPharmacovigilanceCV}}
\end{figure}
 The benefit of taking discreteness into account is evident: the discrete critical values are considerably (by a factor of $2.5 \sim 4$) larger than their respective classical counterparts which leads to more powerful procedures, see also the first row of Table \ref{tab:Empirical data}.

\begin{table}[htbp]
	\caption{Number of rejections  for the pharmacovigilance data.\label{tab:Empirical data}}
	\centering
	\begin{tabular}{lccccc}
		\toprule
		\multicolumn{1}{c}{Procedure controls} &  $\LR$ 	& $\DLR$ & $\GR$  		& $\DPB$ & $\DGR$
		\\
		\midrule
		$\P(\FDP > 5\%) \le 0.5 $  &  23    & 27    & 24    & 29 & 29 \\
		$\P(\FDP > 5\%) \le 0.05 $  &  16    & 21    & 16    & 24 & 24 \\
		\bottomrule
	\end{tabular}
\end{table}

Note that the critical values of  $\DPB$ are not displayed in Figure \ref{fig:FDXPharmacovigilanceCV} since they are visually indistinguishable from the $\DGR$ critical values. Figure \ref{fig:FDXPharmacovigilanceBoxplot} shows that this is in fact true for all indices, thus $\DGR$ is not only an efficient, but also quite accurate approximation of the $\DPB$ values, at least for the discrete distribution involved in this example.
\begin{figure}[htbp]
	\centering
	\makebox{\includegraphics[width=0.7\linewidth]{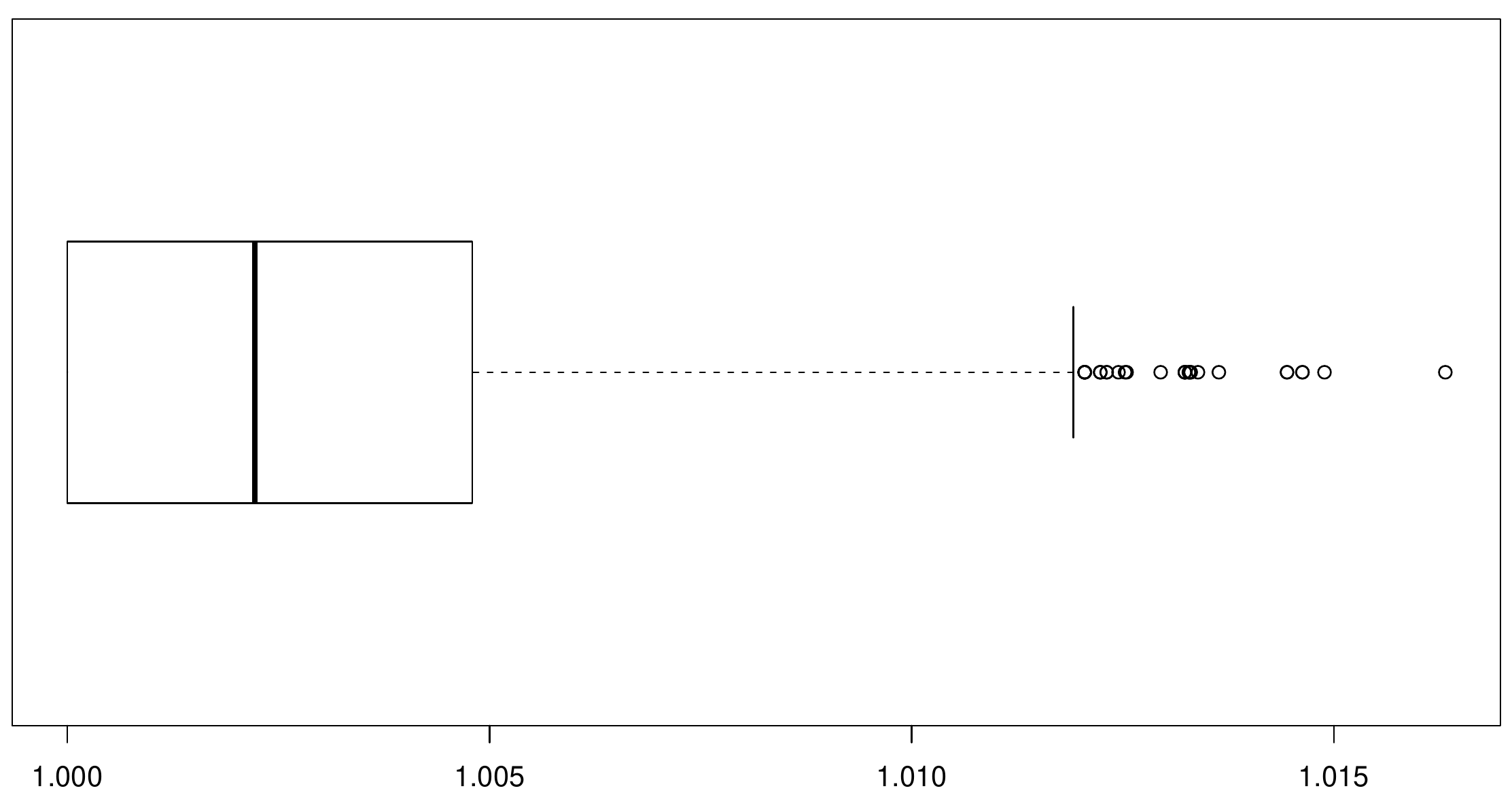}}
	\caption{Boxplot for the ratio of the $\DPB$ to the $\DGR$ critical values.   	\label{fig:FDXPharmacovigilanceBoxplot}}
\end{figure}

We also compare the performance of the above procedures over the full range of possible values for $\zeta$. Figure \ref{fig:FDXPharmacovigilance} depicts the number of rejections  when controlling $\P(\FDP > 5\%) \le \zeta $ for  $\zeta \in (0,1)$.  
\begin{figure}[htbp]
	\centering
	\makebox{\includegraphics[width=1\linewidth]{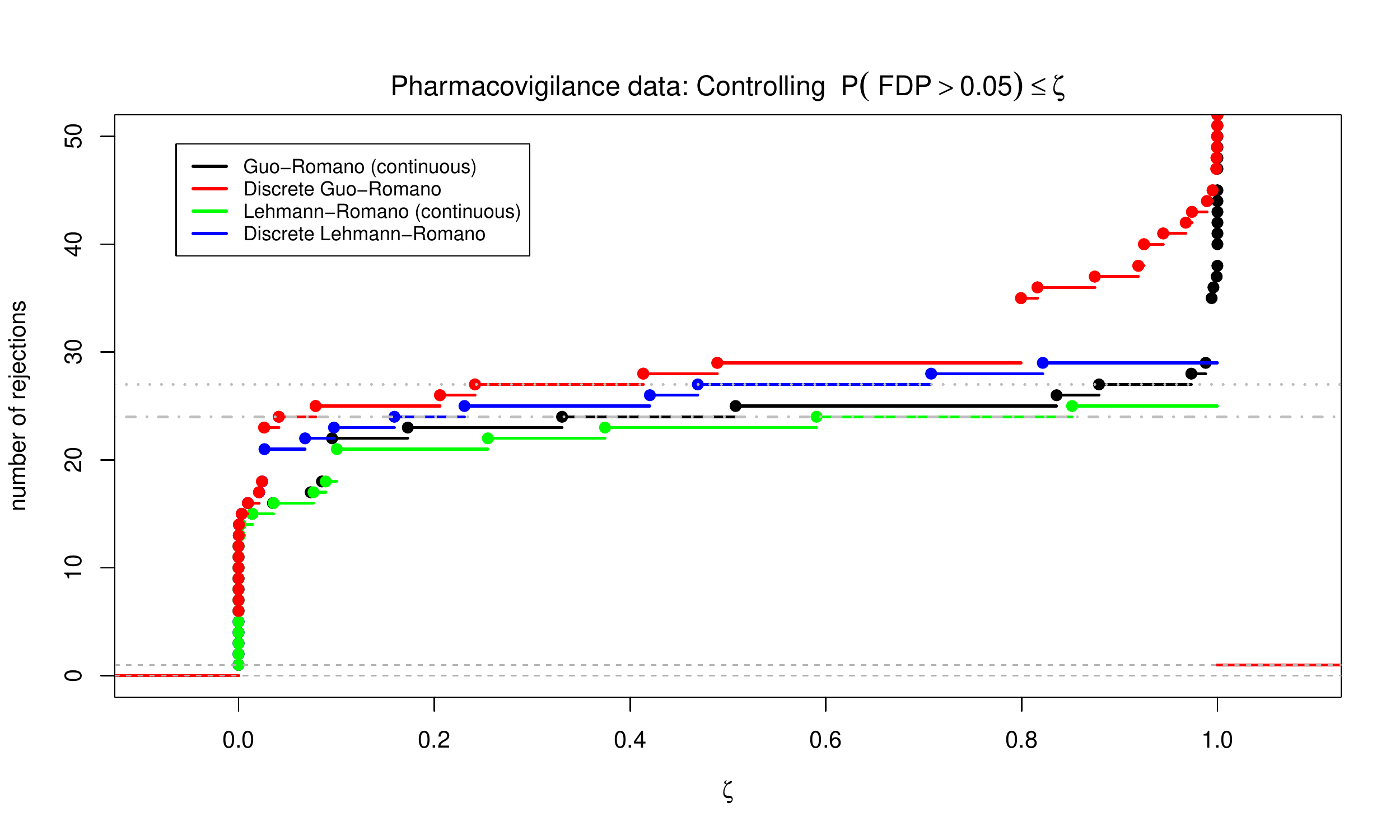}}
	\caption{Number of rejections ($y$-axis) for the pharmacovigilance data when using the $\LR$, $\DLR$, $\GR$ and $\DGR$ procedures. All procedures control the tail probabilities (on the $x$-axis) for the $\FDP$ exceeding $5\%$. The horizontal lines represent the rejections of the BH and discrete BH procedures at $\FDR$-level $5\%$.   	\label{fig:FDXPharmacovigilance}}
\end{figure}
As expected from Propositions~\ref{prop:HLR},~\ref{prop:HPB}~and~\ref{prop:HGR}, the discrete variants reject more hypotheses than their classical counterparts for all values of $\zeta$. For central values of $\zeta$, the gain is about three to four additionally rejected hypotheses, which corresponds roughly to the gain from using the discrete version of BH  instead of $\BH$ (see Table~1 in \cite{DDR2018}).  Figure \ref{fig:FDXPharmacovigilance} also shows that for more extreme values of $\zeta$ the gain may be more pronounced, e.g., when $\P(\FDP > 5\%) \le 0.05 $ is to be guaranteed, the $\GR$ procedure rejects 16 hypotheses, whereas the $\DGR$ procedure rejects 24 hypotheses (see the second row of Table \ref{tab:Empirical data}).

\section{Discussion}\label{sec:discussion}

In this paper, we presented new procedures controlling the FDX while incorporating the (heterogeneous) family of null distribution $\{F_i,1\leq i\leq m\}$. Markedly, it put forward that the geometric averaging of the $F_i$'s is a suitable operation for FDX control. This is new to our knowledge, as all previous works are mostly based on arithmetic averaging of the $F_i$'s (or variation thereof). Maybe more importantly, our approach led to a substantial power improvement in two common situations, under continuity of the tests statistics via weighting schemes, and for discrete test statistics when performing multiple individual Fisher's exact tests. 

This work opens several directions of research. First, the proofs of all our FDX bounds rely on using a kind of independence between the $p$-values (see \eqref{Indep0} and \eqref{Indep}). While this assumption is classical, it is desirable to remove this condition to better stick to the reality of the experiments. This generalization seems however challenging, as FDX control under dependence is already delicate to study in the homogeneous case, see \cite{DR2015}. 
A second interesting avenue is to derive theoretical bounds for the true discovery proportion (TDP) of our procedure. In particular, a useful concern would be to assess whether our way to account for heterogeneity (via arithmetic or geometric averaging of the $F_i$'s) is optimal in some sense. Lastly, our work paves the way to control other simultaneous inference criteria based on an event probability, e.g., to establish post hoc bounds in the discrete heterogeneous case, see \cite{GW2006,GS2011,BNR2019}. While challenging, this is a very exciting direction for future research.

\section*{Acknowledgements}
This work has been supported by ANR-16-CE40-0019 (SansSouci), ANR-17-CE40-0001 (BASICS) and by the GDR ISIS through the "projets exploratoires" program (project TASTY). 
The authors thank Florian Junge for implementing the discrete FDX procedures and improved Poisson-binomial distribution functions in R, and for running the simulations.

\bibliographystyle{apalike}
\bibliography{biblio}

\section{Materials for the proofs}\label{sec:theory}

\subsection{Proving the main tool}
The proof is based on the following result, which is a reformulation of Theorem~5.2 in \cite{Roq2011} in our context.

\begin{theorem}[\cite{Roq2011}]\label{maintool}
In the setting defined in Section~\ref{sec:setting}, consider  any step-down procedure $R$ with critical values $\tau_\l$, $1\leq \l\leq m$.
Then  for all $P\in\mathcal{P}$, we have
\begin{align}
&\FDX(R,P) \nonumber\\
&\leq\:
 \sum_{\ell=1}^m \mathds{1}\{|\cH_0(P)|\leq m(\l)\}\:\:\P_{X\sim P}\left(\sum_{i\in  \cH_0(P)} \mathds{1}\{p_i(X)\leq \tau_\l\} \geq \lfloor\alpha\l\rfloor+1, \wt{\l}(P)=\l\right),\label{keyrelation}
\end{align}
where $\wt{\l}(P)=\min\left\{ \ell \in \{1,\dots,m\} \::\: \l - \sum_{i\in  \cH_1(P)} \mathds{1}\{p_i(X)\leq \tau_\l\} \geq \lfloor\alpha\l\rfloor+1 \right\}$ (with $\wt{\l}(P)=m+1$ if the set is empty).
\end{theorem}

Let us show that Theorem~\ref{maintool} implies \eqref{toolbound} under \eqref{Indep0}. 
 Under \eqref{Indep0}, $\wt{\l}(P)$ is independent of the variable family $$\left\{\sum_{i\in  \cH_0(P)} \mathds{1}\{p_i(X)\leq \tau_\l\},\:\: 1\leq \l\leq m\right\}.$$
Hence, \eqref{keyrelation} provides that $\FDX(R,P)$ is smaller or equal to
\begin{align*}
& \sum_{\ell=1}^m \mathds{1}\{|\cH_0(P)|\leq m(\l)\}\:\:\P_{X\sim P}\left(\sum_{i\in  \cH_0(P)} \mathds{1}\{p_i(X)\leq \tau_\l\} \geq \lfloor\alpha\l\rfloor+1\right) \P_{X\sim P}\left( \wt{\l}(P)=\l\right)\\
&\leq B(\tau,\alpha),
\end{align*}
which gives \eqref{toolbound}.

Finally, for completeness, let us now prove Theorem~\ref{maintool}. 
Let $R_\l=\sum_{i=1}^m \mathds{1}\{p_i(X)\leq \tau_\l\}$ for all $\l$. First, we have for any $\l\in\{1,$\dots$,m\}$ such that $|R_\l|=\l$:
\begin{align}
\{\FDP(R_\l,P)>\alpha\} &= \{|\cH_0(P)\cap R_\l|> \alpha \l\} 
=\{|\cH_0(P)\cap R_\l| \geq  \lfloor\alpha\l\rfloor +1\} \nonumber\\
&=\{ \l - |\cH_1(P)\cap R_\l|  \geq  \lfloor\alpha\l\rfloor +1\}\subset\{\l\geq \wt{\l}(P)\},\nonumber
\end{align}
by using the definition of $\wt{\l}(P)$.
Assuming now $|R_{\l'}|\geq \l'$ for any $\l'\leq \l$, we obtain
\begin{align*}
\{\FDP(R_\l,P)>\alpha\} &\subset \{\l\geq \wt{\l}(P), |R_{\wt{\l}(P)}|\geq \wt{\l}(P)\}\subset \{|\cH_0(P)\cap R_{\wt{\l}(P)}|\geq \lfloor\alpha\wt{\l}(P)\rfloor +1 \},
\end{align*}
where the last step uses the definition of $\wt{\l}(P)$.
Moreover, if $\wt{\l}(P)\geq 2$, again by definition of $\wt{\l}(P)$, we have $(\wt{\l}(P)-1) - \sum_{i\in  \cH_1(P)} \mathds{1}\{p_i(X)\leq \tau_{\wt{\l}(P)-1}\} < \lfloor\alpha(\wt{\l}(P)-1)\rfloor+1$. Hence, we obtain the following upper-bound for $|\cH_0(P)|$: 
$$|\cH_0(P)| = m-  |\cH_1(P)|\leq m-  |\cH_1(P)\cap R_{\wt{\l}(P)-1}| \leq m-\wt{\l}(P)+ \lfloor\alpha(\wt{\l}(P)-1)\rfloor+1.$$
Since the above bound is also true when $\wt{\l}(P)=1$, it holds for any possible value of $\wt{\l}(P)\leq m$. Since $\l=\wh{\l}$ in \eqref{SDrewrite} satisfies both $|R_\l|=\l$ and $|R_{\l'}|\geq \l'$ for any $\l'\leq \l$, combining the above displays gives \eqref{keyrelation}.

\subsection{Proof of Proposition~\ref{prop:equivSD}}
\label{sec:proof:prop:equivSD}

First, we have with $P$-probability $1$, for all $i\in\{1,\dots,m\}$, $p_i\in \mathcal{A}$, both under \eqref{cont} or \eqref{discrete}. Hence, by \eqref{inversecritvalues}, we have $\{\l\in\{1,\dots,m\}\::\:\xi_\l(p_{\sigma(\l)})\leq \zeta\}=\{\l\in\{1,\dots,m\}\::\:p_{\sigma(\l)}\leq \tau_{\l}\}$. By \eqref{SDrewrite}, this gives 
\begin{equation*}
\wh{\l} = \max\{\l\in \{0,\dots,m\}\::\: \forall \l'\leq \l, \: p'_{\l'}\leq \zeta\},
\end{equation*}
where we have denoted $p'_\l=\xi_{\l}(p_{\sigma(\l)})$ for all $\l$.
Now note that
\begin{align*}
\{\sigma(1),\dots,\sigma(\wh{\l})\}&=\{i\in\{1,\dots,m\}\::\: \sigma^{-1}(i)\in \{1,\dots, \wh{\l}\}\}\\
&= \{i\in\{1,\dots,m\}\::\: \forall \l \in \{1,\dots, \sigma^{-1}(i)\}, p'_\l\leq \zeta\}\\
&= \{i\in\{1,\dots,m\}\::\: \max_{\l \in \{1,\dots, \sigma^{-1}(i)\}}\{p'_\l\}\leq \zeta\},
\end{align*}
hence it is sufficient to prove that $\tilde{p}_i= \max_{\l \in \{1,\dots, \sigma^{-1}(i)\}}\{p'_\l\}$ for any $i\in\{1,\dots,m\}$.
For this, let us fix $i\in\{1,\dots,m\}$ and write 
$\{\l\in\{1,\dots,m\}\::\: p_{\sigma(\l)}\leq p_i\}=\{\l\in\{1,\dots,m\}\::\: \l\leq \sigma^{-1}(i)\} \cup A$, for $A= \{\l\in\{1,\dots,m\}\::\: p_{\sigma(\l)}\leq p_i,\l> \sigma^{-1}(i)\} $. This is possible because, by definition, $\l\leq \sigma^{-1}(i)$ implies $p_{\sigma(\l)}\leq p_i$.
Next, for any $\l\in A$, we have both $p_{\sigma(\l)}\leq p_i$ and $p_{\sigma(\l)}\geq p_i$, which entails
$ p_{\sigma(\l)}= p_i$ and thus $\xi_\l(p_{\sigma(\l)}) = \xi_\l(p_{i})$. Since  $\sigma^{-1}(i)\leq \l$ and by the nonincreasing property of $\l\mapsto \xi_\l(p_{i})$, we have $\xi_\l(p_{i}) \leq \xi_{\sigma^{-1}(i)}(p_{i})=\xi_{\sigma^{-1}(i)}(p_{\sigma(\sigma^{-1}(i))})$. 
This gives $p'_\l \leq p'_{\sigma^{-1}(i)}$ for all $\l\in A$. Therefore,
$$
\max_{\substack{1\leq \l\leq m\\p_{\sigma(\l)}\leq p_i}}\{p'_{\l}\} =\max_{\substack{1\leq \l\leq m\\\l\leq \sigma^{-1}(i)}}\{p'_{\l}\}  \vee \max_{\l\in A}\{p'_{\l}\}  
=\max_{\substack{1\leq \l\leq m\\\l\leq \sigma^{-1}(i)}}\{p'_{\l}\},
$$
which leads to the result.

\subsection{An auxiliary lemma}

\begin{lemma}\label{lem:monotoneHGRstar}
With the notation in \eqref{equFjt} the quantity 
\begin{equation}\label{quantityxi:HGRstar}
\P\left( \Bin\left[m-\l+i, \tilde{F}_{m-\l+i}(t) \right] \geq i\right)=\P\left( \Bin\left[m-\l+i, 1-\tilde{F}_{m-\l+i}(t) \right] \leq m-\l\right)
\end{equation} is non-increasing both in $i\in\{1,\dots,\lfloor \alpha\l\rfloor+1\}$ and $\l\in\{1,\dots,m\}$.
\end{lemma}

\begin{proof}
First note that 
$$1-\tilde{F}_{j}(t)=\left(\prod_{j'=1}^{j} (1-(F(t))_{(j')})\right)^{1/j}$$
 in non-decreasing in $j$ (because the geometric average of larger numbers is larger).  The quantity \eqref{quantityxi:HGRstar} is thus non-increasing with respect to $i$, so that the only thing to check is that this quantity is non-increasing with respect to $\l$. For this, it is sufficient to prove that 
$
\Bin\left[j+1, \tilde{F}_{j+1}(t) \right]
$
is stochastically larger than 
$
\Bin\left[j, \tilde{F}_{j}(t) \right]
$
for any $j\in\{1,\dots,m-1\}$ (which is not obvious because $\tilde{F}_{j}(t)\geq \tilde{F}_{j+1}(t)$).
Let $n_1=j$,  $p_1 = \tilde{F}_{j}(t)$, $n_2=1$, $p_2 = (F(t))_{(j+1)}$, $n=j+1$ and $p=\tilde{F}_{j+1}(t)$.
We easily check that $n=n_1+n_2$ and by \eqref{equFjt}, 
\begin{align*}
(1-p)^n &= \prod_{j'=1}^{j+1} (1-(F(t))_{(j')})\\
&= \prod_{j'=1}^{j} (1-(F(t))_{(j')})\times (1-(F(t))_{(j+1)})=(1-p_1)^{n_1} (1-p_2)^{n_2}.
\end{align*}
Applying Example 1.A.25 in \cite{Shaked} ($m=2$ with the notation therein), we obtain that 
the sum of a
$
\Bin\left[n_1, p_1 \right]
$
variable 
and a 
$
\Bin\left[n_2, p_2 \right]
$
variable (with independence)
is stochastically smaller than a $\Bin\left[n, p \right]$ variable.
In particular, a $
\Bin\left[n_1, p_1 \right]
$ variable
is stochastically smaller than a $\Bin\left[n, p \right]$ variable. This gives the result.
\end{proof}

\appendix

\section{Additional numerical details}\label{sec:appendix}

\begin{table}[htbp]
 \centering
\small
\caption{Average power (i.e. average of true discovery proportion) of FDX controlling procedures (at $\zeta=0.5$) for $N=25$.}
{\tiny
\begin{tabular}{cccc||cccccc}
	\toprule
	$m$     & $m_3$    & $m_1$    & $q_3$    & [BH]    & [LR]    & [DLR]   & [GR]    & [DPB]  &  [DGR] \\\midrule
		800   & 80    & 144   & 0.15  & 0     & 0     & 0.0025 & 0.0002 & 0.0025 & 0.0025 \\
		&       & 144   & 0.25  & 0.0004 & 0     & 0.043 & 0.0077 & 0.043 & 0.043 \\
		&       & 144   & 0.4   & 0.0803 & 0     & 0.3328 & 0.1195 & 0.4412 & 0.4406 \\
		&       & 360   & 0.15  & 0     & 0     & 0.0025 & 0.0002 & 0.0043 & 0.0043 \\
		&       & 360   & 0.25  & 0.0004 & 0     & 0.043 & 0.0077 & 0.0444 & 0.0444 \\
		&       & 360   & 0.4   & 0.0803 & 0     & 0.3766 & 0.1195 & 0.4512 & 0.4511 \\
		&       & 576   & 0.15  & 0     & 0     & 0.0071 & 0.0002 & 0.0076 & 0.0076 \\
		&       & 576   & 0.25  & 0.0004 & 0     & 0.0528 & 0.0077 & 0.077 & 0.077 \\
		&       & 576   & 0.4   & 0.0803 & 0     & 0.4474 & 0.1195 & 0.5141 & 0.5128 \\
		& 240   & 112   & 0.15  & 0     & 0     & 0.0025 & 0.0002 & 0.0025 & 0.0025 \\
		&       & 112   & 0.25  & 0.0005 & 0     & 0.0289 & 0.0076 & 0.0422 & 0.0422 \\
		&       & 112   & 0.4   & 0.2148 & 0     & 0.425 & 0.1984 & 0.5153 & 0.5139 \\
		&       & 280   & 0.15  & 0     & 0     & 0.0025 & 0.0002 & 0.0025 & 0.0025 \\
		&       & 280   & 0.25  & 0.0005 & 0     & 0.0336 & 0.0076 & 0.0429 & 0.0429 \\
		&       & 280   & 0.4   & 0.2147 & 0     & 0.4413 & 0.1983 & 0.5728 & 0.5716 \\
		&       & 448   & 0.15  & 0     & 0     & 0.0025 & 0.0002 & 0.0037 & 0.0037 \\
		&       & 448   & 0.25  & 0.0005 & 0     & 0.0389 & 0.0076 & 0.043 & 0.043 \\
		&       & 448   & 0.4   & 0.2145 & 0     & 0.4609 & 0.1983 & 0.5921 & 0.5917 \\
		& 640   & 32    & 0.15  & 0     & 0     & 0.0018 & 0.0002 & 0.0025 & 0.0025 \\
		&       & 32    & 0.25  & 0.001 & 0.0003 & 0.0203 & 0.0075 & 0.0212 & 0.0212 \\
		&       & 32    & 0.4   & 0.4243 & 0.0203 & 0.4908 & 0.5379 & 0.673 & 0.6724 \\
		&       & 80    & 0.15  & 0     & 0     & 0.002 & 0.0002 & 0.0025 & 0.0025 \\
		&       & 80    & 0.25  & 0.001 & 0.0003 & 0.0203 & 0.0075 & 0.0212 & 0.0212 \\
		&       & 80    & 0.4   & 0.4242 & 0.0203 & 0.4974 & 0.5374 & 0.6746 & 0.6743 \\
		&       & 128   & 0.15  & 0     & 0     & 0.0021 & 0.0002 & 0.0025 & 0.0025 \\
		&       & 128   & 0.25  & 0.001 & 0.0003 & 0.0203 & 0.0075 & 0.0212 & 0.0212 \\
		&       & 128   & 0.4   & 0.424 & 0.0203 & 0.5048 & 0.5369 & 0.6753 & 0.675 \\
		2000  & 200   & 360   & 0.15  & 0     & 0     & 0.0007 & 0     & 0.0022 & 0.0022 \\
		&       & 360   & 0.25  & 0.0001 & 0     & 0.0198 & 0.0029 & 0.0222 & 0.0222 \\
		&       & 360   & 0.4   & 0.073 & 0     & 0.3331 & 0.0792 & 0.4315 & 0.4311 \\
		&       & 900   & 0.15  & 0     & 0     & 0.0022 & 0     & 0.0024 & 0.0024 \\
		&       & 900   & 0.25  & 0.0001 & 0     & 0.021 & 0.0029 & 0.0373 & 0.0373 \\
		&       & 900   & 0.4   & 0.073 & 0     & 0.338 & 0.0792 & 0.4515 & 0.4515 \\
		&       & 1440  & 0.15  & 0     & 0     & 0.0024 & 0     & 0.0024 & 0.0024 \\
		&       & 1440  & 0.25  & 0.0001 & 0     & 0.0378 & 0.0029 & 0.0428 & 0.0428 \\
		&       & 1440  & 0.4   & 0.0729 & 0     & 0.432 & 0.0792 & 0.5173 & 0.5144 \\
		& 600   & 280   & 0.15  & 0     & 0     & 0.0007 & 0     & 0.0007 & 0.0007 \\
		&       & 280   & 0.25  & 0.0001 & 0     & 0.0197 & 0.0029 & 0.0205 & 0.0205 \\
		&       & 280   & 0.4   & 0.2058 & 0     & 0.4093 & 0.196 & 0.5194 & 0.5176 \\
		&       & 700   & 0.15  & 0     & 0     & 0.0007 & 0     & 0.002 & 0.002 \\
		&       & 700   & 0.25  & 0.0001 & 0     & 0.02  & 0.0029 & 0.0205 & 0.0205 \\
		&       & 700   & 0.4   & 0.2058 & 0     & 0.4374 & 0.196 & 0.5678 & 0.5657 \\
		&       & 1120  & 0.15  & 0     & 0     & 0.0014 & 0     & 0.0024 & 0.0024 \\
		&       & 1120  & 0.25  & 0.0001 & 0     & 0.0201 & 0.0029 & 0.0206 & 0.0206 \\
		&       & 1120  & 0.4   & 0.2057 & 0     & 0.4545 & 0.1959 & 0.5908 & 0.5906 \\
		& 1600  & 80    & 0.15  & 0     & 0     & 0.0007 & 0     & 0.0007 & 0.0007 \\
		&       & 80    & 0.25  & 0.0003 & 0.0001 & 0.009 & 0.0029 & 0.0172 & 0.0172 \\
		&       & 80    & 0.4   & 0.4223 & 0.0114 & 0.4823 & 0.5288 & 0.6665 & 0.6658 \\
		&       & 200   & 0.15  & 0     & 0     & 0.0007 & 0     & 0.0007 & 0.0007 \\
		&       & 200   & 0.25  & 0.0003 & 0.0001 & 0.009 & 0.0029 & 0.0184 & 0.0184 \\
		&       & 200   & 0.4   & 0.4222 & 0.0114 & 0.4866 & 0.5286 & 0.6689 & 0.6683 \\
		&       & 320   & 0.15  & 0     & 0     & 0.0007 & 0     & 0.0007 & 0.0007 \\
		&       & 320   & 0.25  & 0.0003 & 0.0001 & 0.009 & 0.0029 & 0.0194 & 0.0194 \\
		&       & 320   & 0.4   & 0.422 & 0.0114 & 0.4935 & 0.5283 & 0.6724 & 0.6715 \\
	\end{tabular}%
}
	\label{tab:addlabel}%
\end{table}%

\end{document}